\RequirePackage{fix-cm}
\documentclass[final]{svjour3}

\smartqed  
\usepackage[T1]{fontenc}
\usepackage[utf8]{inputenc}
\usepackage{graphicx}
\usepackage{cite} 
\usepackage{xcolor,stmaryrd}
\usepackage{url}
\usepackage{soul}
\usepackage[notransparent]{svg}
\usepackage{doi}
\usepackage{makecell}
\usepackage{multirow}
\usepackage{physics,theorem}
\usepackage{mathrsfs,dsfont,amsfonts,mathtools,slashed,mathtools}
\makeatletter
\let\cl@chapter\undefined
\makeatletter
\usepackage{booktabs}
\usepackage{xfrac}
\usepackage{hyperref}
\usepackage[capitalize]{cleveref}

\crefname{section}{Sect.}{Sects.}
\Crefname{section}{Section}{Sections}
\crefname{equation}{Eq.}{Eqs.}
\Crefname{equation}{Equation}{Equations}
\crefname{figure}{Fig.}{Figs.}
\Crefname{figure}{Figure}{Figures}
\crefname{definition}{Def.}{Defs.}
\Crefname{definition}{Definition}{Definitions}
\crefname{theorem}{Thm.}{Thms.}
\Crefname{theorem}{Theorem}{Theorems}
\crefname{innercustomthm}{Thm.}{Thms.}
\Crefname{innercustomthm}{Theorem}{Theorems}
\usepackage{color}

\newtheorem{innercustomthm}{Theorem}
\newenvironment{customthm}[1]
{\renewcommand\theinnercustomthm{#1}\innercustomthm}
{\endinnercustomthm}

\newtheorem{innercustomresult}{Result}

\definecolor{darkblue}{RGB}{0,0,127} 
\definecolor{darkgreen}{RGB}{0,180,190}
\definecolor{darkred}{RGB}{200, 10, 42}
\hypersetup{
	colorlinks, 
	linkcolor=darkred, 
	citecolor=darkblue, 
	filecolor=darkred, 
	urlcolor=darkblue,
	pdftitle={Free fermions behind the disguise}, 
	pdfauthor={Samuel J. Elman, Adrian Chapman, Steven T. Flammia}
}

\renewcommand{\epsilon}{\varepsilon}

\newcommand{\psiv}{\psi^{\vphantom{\dagger}}}
\newcommand{\psid}{\psi^\dagger}
\usepackage[hmargin=3cm,bmargin=3.5cm]{geometry}
\def\makeheadbox{\relax} 

\begin{document}
	
	\title{Free fermions behind the disguise}
	
	\author{Samuel J.~Elman \and Adrian Chapman \and Steven~T.~Flammia}
	
	\authorrunning{S.\ J.\ Elman, A.\ Chapman, and S.\ T.\ Flammia}
	
	\institute{S.\ J.\ Elman \and A.\ Chapman \at
		Centre for Engineered Quantum Systems, School of Physics, The University of Sydney, Sydney, Australia \\
		\email{\{samuel.elman, adrian.chapman\}@sydney.edu.au}
		\and
		S.\ T.\ Flammia \at
		AWS Center for Quantum Computing, Pasadena, CA, 91125, USA \\
		\email{sflammi@amazon.com}
	}
	
	\date{} 
	
	\maketitle
	
	\begin{abstract}
		An invaluable method for probing the physics of a quantum many-body spin system is a mapping to noninteracting effective fermions.
		We find such mappings using only the frustration graph $G$ of a Hamiltonian $H$, i.e., the network of anticommutation relations between the Pauli terms in $H$ in a given basis.
		Specifically, when $G$ is (even-hole, claw)-free, we construct an explicit free-fermion solution for $H$ using only this structure of $G$, even when no Jordan-Wigner transformation exists. 
		The solution method is generic in that it applies for any values of the couplings. 
		This mapping generalizes both the classic Lieb-Schultz-Mattis solution of the XY model and an exact solution of a spin chain recently given by Fendley, dubbed ``free fermions in disguise.'' 
		Like Fendley's original example, the free-fermion operators that solve the model are generally highly nonlinear and nonlocal, but can nonetheless be found explicitly using a transfer operator defined in terms of the independent sets of $G$. 
		The associated single-particle energies are calculated using the roots of the independence polynomial of $G$, which are guaranteed to be real by a result of Chudnovsky and Seymour. 
		Furthermore, recognizing (even-hole, claw)-free graphs can be done in polynomial time, so recognizing when a spin model is solvable in this way is efficient.
		In a crucial step to proving our result, we additionally prove that there exists a hierarchy of commuting conserved charges for models whose frustration graphs are claw-free only, and hence these models are integrable.
		Finally, we give several example families of solvable and integrable models for which no Jordan-Wigner solution exists, and we give a detailed analysis of such a spin chain having 4-body couplings using this method. 
	\end{abstract}
	
	\section{Introduction}
	\label{sec:intro}
	A notorious challenge for the simulation of quantum many-body systems is the exponential growth of the Hilbert space dimension in the number of constituent degrees of freedom.
	Systems for which this difficulty can be circumvented via an analytic solution are invaluable for at least two reasons.
	First, the discovery of a new class of analytic solutions opens up the prospect of tractable simulation to a new family of models, and potentially of new phenomenology.
	Second, analytic solutions can be taken as starting points for approximations to more realistic models, thus extending the reach of these methods. 
	
	For a quantum spin-$\sfrac{1}{2}$ (qubit) system, one remarkable type of analytic solution comes in the form of a duality to effective fermions.
	When the effective fermions are noninteracting, it can be said that we have found a means of restricting the model's essential behavior to the low-dimensional subspace of a single fermion, and the physics of the model is well-understood.
	The textbook example of a free-fermion mapping is the Jordan-Wigner transformation~\cite{jordan1928uber}, which was famously employed to solve the one-dimensional XY model by Lieb, Schultz, and Mattis \cite{lieb1961two}. 
	The key insight is the identification of nonlocal Pauli operators with fermionic ladder operators.
	In the fermionic picture, the $n$-qubit XY-model Hamiltonian is mapped to a free-fermion Hamiltonian on $2n$ fermionic modes.
	The model is then completely solved by exactly diagonalizing the model's $2n \times 2n$ free-fermion Hamiltonian, an exponential simplification from the naive brute-force diagonalization that one might expect to need in the qubit picture.
	This solution method is \textit{generic}, meaning it applies regardless of the values taken by the nonzero coupling constants in the Hamiltonian. 
	This is because the Jordan-Wigner transformation maps each term in the Hamiltonian linearly to a fermionic bilinear operator. 
	
	One physical signature of models that are solvable in this way is an energy spectrum $\{E_{\boldsymbol{x}}\}$ that is given in terms of a number $\alpha$ of single-particle energies $\epsilon_k$ by
	\begin{align}
	E_{\boldsymbol{x}}= \sum_{k = 1}^{\alpha} (-1)^{x_k} \epsilon_k\,,
	\label{eq:free_energy}
	\end{align}
	where $\alpha\le n$ and $\boldsymbol{x} \in \{0, 1\}^{\times \alpha}$ is a binary vector describing the occupation of each canonical fermionic mode. 
	We will refer to a spectrum of the form in~\cref{eq:free_energy} as \emph{free}. 
	We say that a Hamiltonian is \textit{solvable} if it has a free spectrum and it can be written in terms of its eigenmodes $\psi_j$ as
	\begin{align}
	H = \sum_{k = 1}^{\alpha} \epsilon_{k} [\psiv_k,\psid_k] \,,
	\label{eq:ffsolvable}
	\end{align}
	where the $\psiv_k$ obey the \textit{canonical anticommutation relations}, $\{\psiv_j, \psiv_k\} = 0$ and $\{\psiv_j,\psi_k^\dagger\} = \delta_{jk} I$. 
	When $\alpha < n$, a free spectrum will necessarily have degeneracies, since this is equivalent to the case where a subset of the energies $\{\epsilon_k\}_{k = 1}^n$ are equal to zero. 
	
	\begin{table*}
		\centering
		\setcellgapes{3pt}
		\makegapedcells
		\begin{tabular} {c c c c}
			\toprule
			\textit{Includes} & \multicolumn{2}{c}{\textit{Forbidden}} & \\  \cmidrule{2-4} 
			Simplicial Clique & Claw, $K_{1, 3}$ & Even hole, $C_{2k}$ & \\ \midrule
			\makecell{(a)\vspace{16mm}} \makecell{\includegraphics[width=0.18 \textwidth]{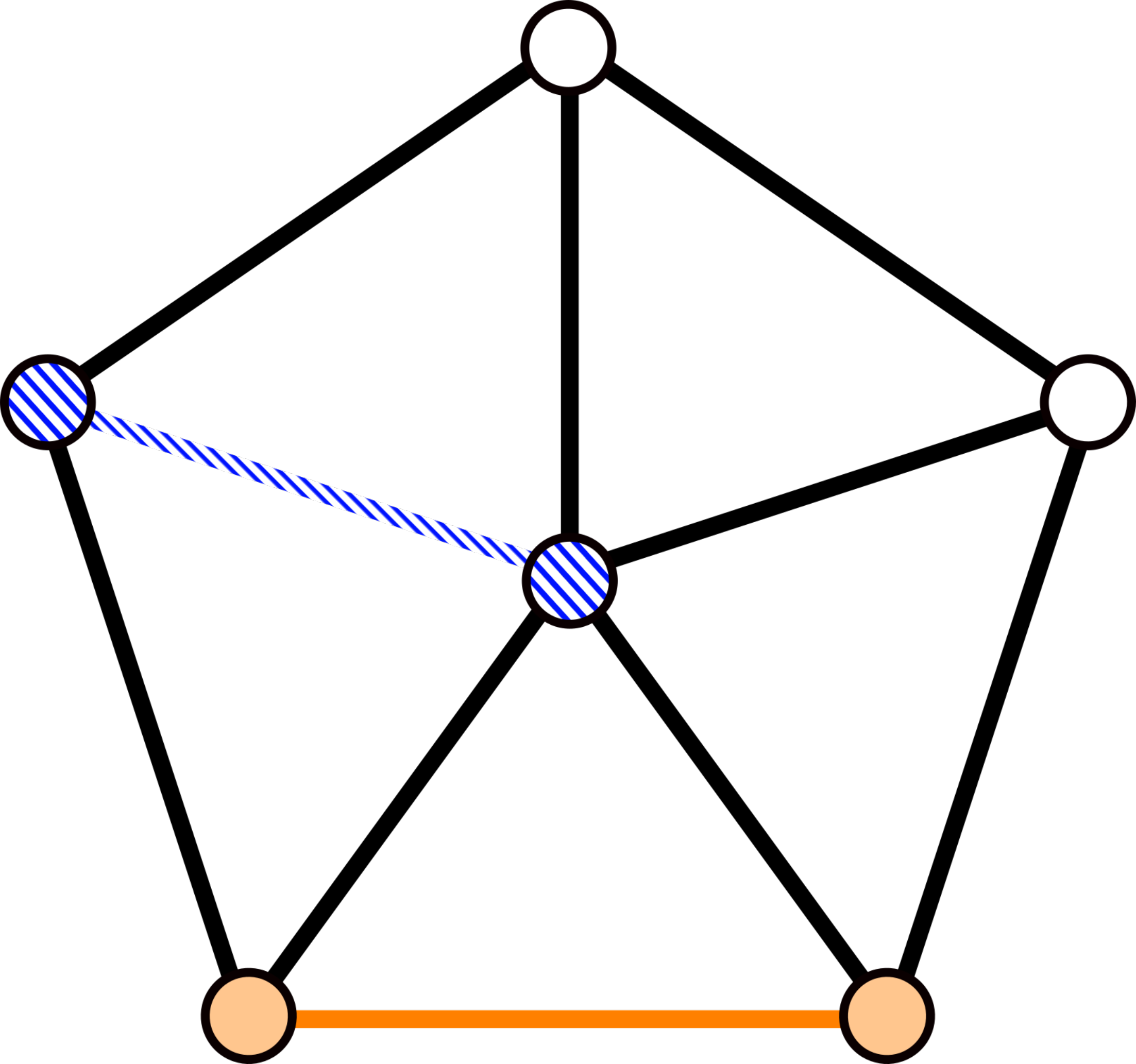}} &
			\makecell{(b)\vspace{16mm}} \makecell{\includegraphics[width=0.18 \textwidth]{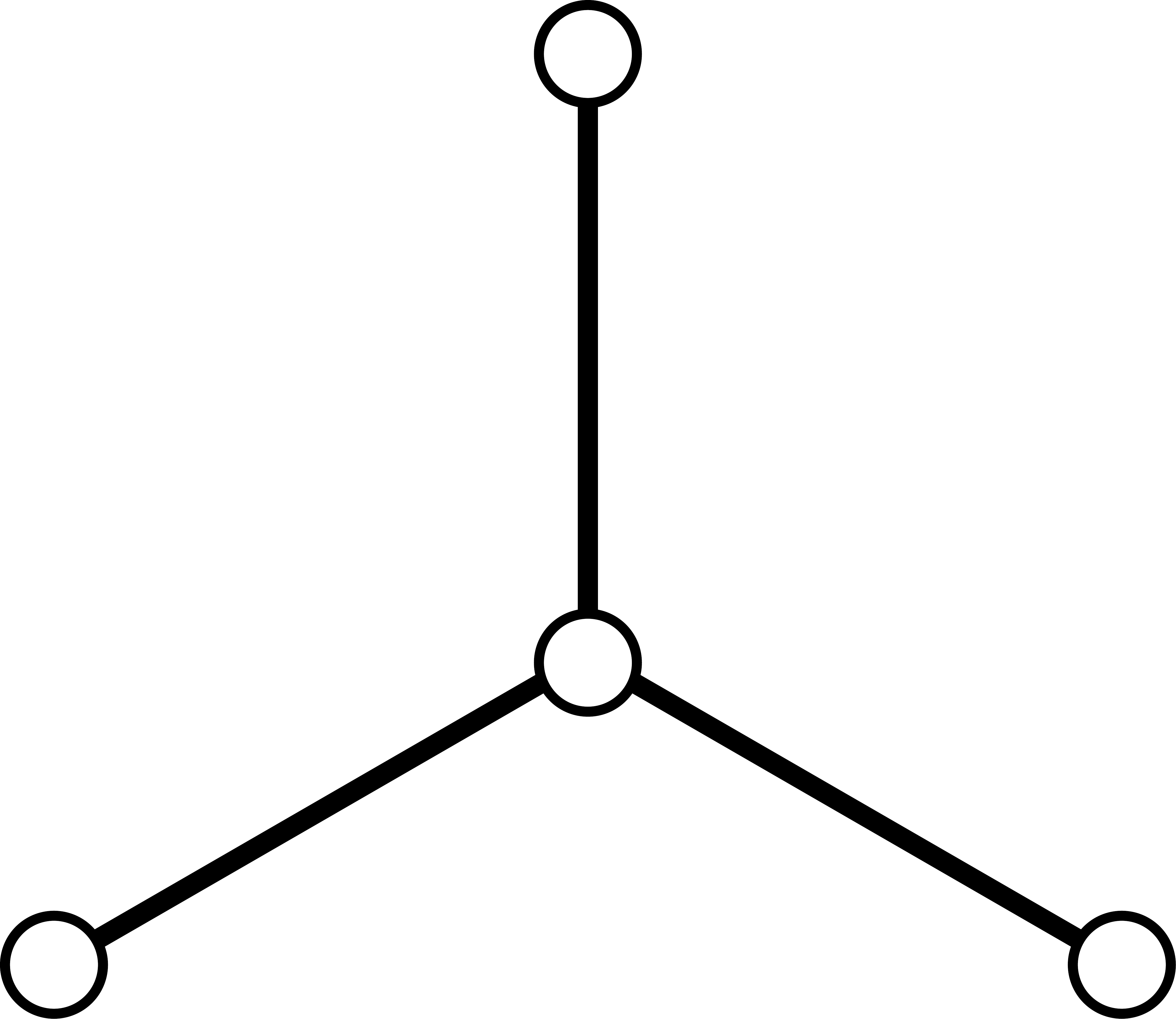}} & \makecell{(c)\vspace{16mm}} \makecell{\includegraphics[width=0.39\textwidth]{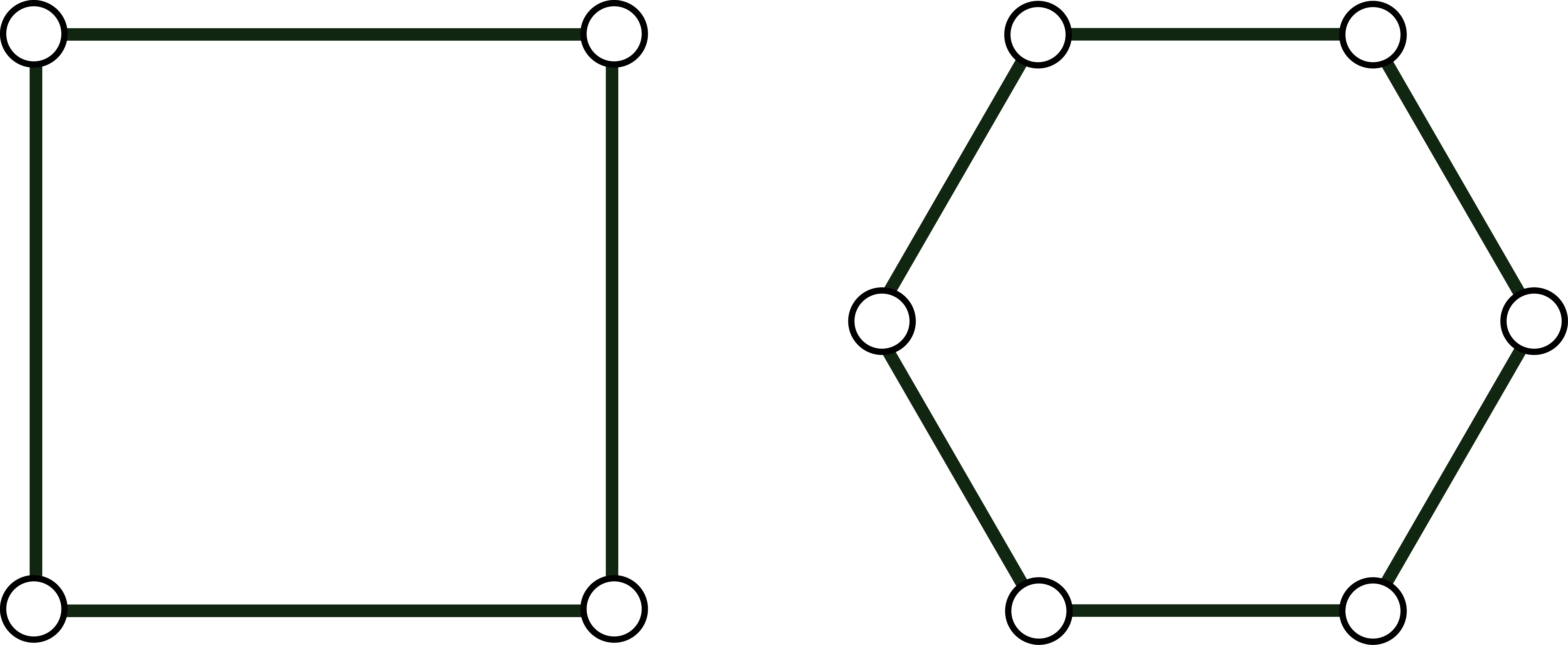}} & $\mbox{\Large \dots}$ \\ & & \mbox{\hspace{0.015\textwidth}} $k = 2$ \mbox{\hspace{0.15\textwidth}} $k = 3$ & \\ \bottomrule
		\end{tabular}
		\caption{
			(a) A graph has a simplicial clique (orange in the example) if it has a clique where the neighborhood of each constituent vertex, minus the original clique, induces another clique (see \cref{def:simplicial}). 
			The blue, hatched nodes are the induced neighbors of the left orange vertex, and they induce a clique (also hatched). 
			The induced neighbors of the right orange vertex similarly induce a clique, so this graph is simplicial. 
			A graph is claw-free or even-hole-free if none of its vertex-induced subgraphs contain (b) the ``claw'' $K_{1,3}$ or (c) an even hole $C_{2k}$. 
			If a graph is (even-hole, claw)-free, it necessarily contains a simplicial clique~\cite{CSSS}.}
		\label{tab:forbidden}
	\end{table*}

	Recently, two of the authors~\cite{Chapman2020Characterization} have given a simple necessary and sufficient condition for a qubit Hamiltonian $H$ to be solvable via a Jordan-Wigner mapping by looking at properties of the \textit{frustration graph} of $H$ (see \cref{def:frustration}). 
	This gives a complete solution for this best-known class of solvable models. 
	However, there exist models that are free and solvable, but that cannot be solved via any Jordan-Wigner mapping. 
	Such a model has been introduced by Fendley~\cite{fendley2019free} as free fermions ``in disguise''. 
	Fendley solves this model by directly defining the single-particle eigenstates for the Hamiltonian through its interaction terms. 
	The free fermions manifest nonlinearly and nonlocally in a basis which is dependent on the specific interaction strengths, but they remain free for all values of the couplings. 
	The solution is therefore generic. 
	This solution method has since been reproduced in a family of generalized spin-chain models~\cite{Alcaraz2020Free, Alcaraz2020Integrable}, including \emph{qudit} models, where the system is dual to so-called free parafermions~\cite{Fendley2014Free}. 
	The generic nature of the free spectrum in these models would suggest that the frustration-graph formalism of Ref.~\cite{Chapman2020Characterization} could be applied to understand this solution. 
	However, since these models provably do not admit a Jordan-Wigner mapping, it would seem the solution relies on much subtler commutation structures than previously understood. 
	
	In this work, we go behind the disguise and clarify the role that the frustration graph plays in solving these models. 
	We develop an infinite family of free-fermion solutions which generalizes Ref.~\cite{fendley2019free} by finding specific graph-theoretic conditions for when such a solution is possible. 
	Specifically, when the frustration graph of $H$ avoids certain obstructions known as claws and even holes (see \cref{tab:forbidden} and \cref{def:ECF}), then $H$ is solvable. 
	
	\begin{innercustomresult}{\rm \bf (Informal version of \cref{thm:maintheorem1} and \cref{thm:maintheorem2}.)} 
		If a Hamiltonian has an (even-hole, claw)-free frustration graph, then it has an explicit free-fermion solution.
	\end{innercustomresult}
	As a corollary to our result, we prove that independent sets in the frustration graph give us a family of conserved charges for the model whenever the frustration graph is only claw-free. That is, that when the frustration graph of $H$ avoids claws, the model is integrable.
	
	\begin{innercustomresult}{\rm \bf (Corollary of \cref{lem:conserved_charges}.)}
			If a Hamiltonian has a claw-free frustration graph, then the model admits a hierarchy of commuting conserved charges.
			\label{result2}
	\end{innercustomresult}
	
	The proof proceeds by considering the independence polynomial of the frustration graph. 
	The independence polynomial of a graph is the polynomial generating function that counts the independent sets in the graph. 
	We can incorporate detailed information about the Hamiltonian into the independence polynomial by attaching certain vertex weights given in terms of squares of Hamiltonian coupling strengths.
	We first prove \cref{result2} under the less restrictive assumption that the frustration graph is claw-free.  
	We then prove that when the frustration graph is additionally even-hole-free, the independence polynomial factorizes into a quadratic function of a certain transfer operator. 
	The single-particle spectrum can then be derived by looking at the zeros of the vertex-weighted independence polynomial. 
	Given knowledge of the spectrum, the transfer operator then acts like a raising operator when acting on a fiducial mode whose existence is guaranteed by the structure of (even-hole, claw)-free graphs. 
	The modes generated in this way allow us to define the eigenmodes of the free-fermionic Hamiltonian. 
	
	As in Fendley's original model~\cite{fendley2019free}, there are no ``physical modes'' to speak of from this derivation as there would be from a Jordan-Wigner transformation. 
	The eigenmodes are ``disguised'', and then emerge directly as nonlinear, nonlocal combinations of the underlying spin operators. 
	We call the eigenmodes revealed by this procedure \textit{incognito modes}. 
	
	We then describe explicit families of models with frustration graphs that satisfy the (even-hole, claw)-free condition. 
	The first of these examples is a small system, which is chosen explicitly as it has no generalized Jordan-Wigner solution and yet has a free-fermion solution of the form we consider. 
	We show that this particular model is in fact related to one with a Jordan-Wigner solution by a local rotation. 
	We next demonstrate how the family of generalized spin chain models included in Refs.~\cite{fendley2019free, Alcaraz2020Free, Alcaraz2020Integrable} fit into our formalism. 
	These models have a particular one-dimensional structure which makes their asymptotic dispersion relations amenable to the techniques of Toeplitz-matrix analysis. 
	Though this is not expected to be true of general (even-hole, claw)-free graphs, a structure theorem for these graphs demonstrates that we should expect their coarse topology to be one-dimensional, or possibly treelike.
	
	The structure of the paper is as follows: in \Cref{sec:informal}, we formally state all definitions and our main results. 
	In \Cref{sec:symmetries}, we discuss our result in the context of prior work. 
	In \Cref{sec:technical} we prove the main results as \cref{thm:maintheorem1} and \cref{thm:maintheorem2}. 
	In \Cref{sec:examples}, we demonstrate how specific examples fit into our formalism, and in \Cref{sec:discussion} we close with a discussion of possible future work.

	\section{Main Results}
	\label{sec:informal}
	We begin with self-contained statements of our main theorems and necessary supporting definitions. 
	First, let us fix graph-theoretic conventions. 
	A graph $G \coloneqq (V, E)$ consists of a set of vertices $V$ together with a set of 2-element subsets $E \subset V^{\times 2}$ called edges. 
	All graphs we consider are finite and simple: every pair of vertices neighbors by at most one edge, and the graphs contain no self loops. 
	An \textit{induced subgraph} is a graph $G[S] \coloneqq (S, E \cap S^{\times2})$ whose vertex set is $S\subseteq V$ and whose edge set consists of all of the edges in $E$ with both endpoints in $S$. 
	We denote the vertex and edge sets of $G[S]$ by $V_S$ and $E_S$, respectively. 
	We will also use the notation $G-W \coloneqq G[V\backslash W]$ to denote the induced subgraph of the graph $G$ by removing the set of vertices $W$. 
	We will often refer to a subset of vertices interchangeably with the subgraph it induces.
	The open neighborhood of a vertex, $\Gamma(v)$ is the set of vertices neighboring the vertex, $v$.
	The closed neighborhood of a vertex, $\Gamma[v]$, is the set of vertices neighboring the vertex, $v$, together with $v$ itself. 
	An \textit{independent set}, or \textit{stable set}, of a graph $G=(V,E)$ is a subset of vertices $S \subseteq V$ which induces a subgraph with no edges, $G[S] = (S, \{\})$.
	Notice that our definition includes the empty set and sets of one vertex as independent sets.

	A clique is a graph where every pair of vertices is neighboring. 
	For us, a particularly important type of clique in a graph is a simplicial clique  (See~\cref{tab:forbidden} (a)):
	\begin{definition}[Simplicial clique]
		\label{def:simplicial}
		A \textit{simplicial clique} $K_s$ in $G$ is a non-empty clique such that for every vertex, $v\in K_s$, the (closed or open) neighborhood of $v$ induces a clique in $G - K_s$. 
		That is, for each $v\in K_s$ we have that $K_v \coloneqq \Gamma[v] \backslash (K_s\backslash v)$ induces a clique in $G$.
	\end{definition}
	The claw, $K_{1, 3}$, is the complete bipartite graph between one vertex and a set of three non-neighboring vertices (See~\cref{tab:forbidden} (b)). 
	A path of length $\ell$ is a connected graph of $\ell + 1$ vertices and $\ell$ edges such that every vertex has at most two neighbors.
	A cycle, $C_{\ell}$, is a connected graph of $\ell$ edges and $\ell$ vertices such that each vertex has exactly two neighbors.
	Informally, a path of length $\ell$ is a cycle $C_{\ell + 1}$ with one edge removed.
	A hole in a graph $G$ is a subset of $\ell$ vertices $W \subseteq V$ such that $G[W] \cong C_{\ell}$ (i.e. an induced cycle of $G$), where $\ell \geq 4$. 
	A hole is called \textit{even} if it has an even number of vertices and edges.
	Importantly, our definition of an even hole includes holes of four vertices (See~\cref{tab:forbidden} (c)).
	
	Next we turn to definitions involving a physical many-body qubit model. 
	Consider an $n$-qubit Hamiltonian, $H$, written in a given basis of Pauli operators $\{\sigma^{\boldsymbol{j}}\}$ as
	\begin{align}
	H \coloneqq \sum_{\boldsymbol{j} \in V} h_{\boldsymbol{j}} \coloneqq \sum_{\boldsymbol{j} \in V}b_{\boldsymbol{j}}\sigma^{\boldsymbol{j}} \,,
	\label{eq:paulihdef}
	\end{align}
	where $V\subseteq \{0,x,y,z\}^{\otimes n}$ is a set of strings labeling the $n$-qubit Pauli operators in the natural way. 
	A frustration graph describes the commutation relations between the Hamiltonian terms as follows:
	
	\begin{definition}[Frustration graph]
		\label{def:frustration}
		The \textit{frustration graph} of a Hamiltonian of the form in \cref{eq:paulihdef} is a graph, $G(H)=(V,E)$, with vertices in $V$ in one-to-one correspondence with the Pauli terms $\{\sigma^{\boldsymbol{j}}\}_{\boldsymbol{j} \in V}$ in $H$, and edge set $E$ defined by the commutation relations between the Hamiltonian terms:
		\begin{equation}
		E=\left\{(\boldsymbol{j},\boldsymbol{k})\big|\{\sigma^{\boldsymbol{j}}, \sigma^{\boldsymbol{k}}\}=0\right\}.
		\end{equation}
		That is, two vertices in $V$ are adjacent in $G(H)$ if and only if their corresponding Paulis anticommute. 
	\end{definition}
	The frustration graph is the complement of the Pauli graph introduced by Planat~\cite{Planat2007Pauli}.
	Notice that it is always simple by construction. Where clear from context, we will drop the dependence on the Hamiltonian from $G(H)$.
	
	\begin{definition}[ECF]
		\label{def:ECF}
		A graph $G$ is said to be \textit{(even-hole, claw)-free}, or \textit{ECF}, if it contains no even holes and no claws among its induced subgraphs (see \cref{tab:forbidden}). 
		A Hamiltonian $H$ is ECF if its frustration graph $G(H)$ is ECF. 
	\end{definition}
	It can be shown that all (even-hole, claw)-free graphs are \textit{simplicial}, meaning they contain a simplicial clique~\cite{CSSS}. 
	If a Hamiltonian $H$ is ECF then its frustration graph is necessarily simplicial, so we say that $H$ is simplicial as well. 
	
	Our first main result says that the spectrum of an ECF Hamiltonian is free, with single-particle energies given by the roots of a certain polynomial.
	\begin{theorem}
		Every ECF Hamiltonian $H$ has a free spectrum of the form in \cref{eq:free_energy}. 
		In particular, the single-particle energies $\{\epsilon_j\}_{j=1}^{\alpha(G)}$ satisfy
		\begin{align}
		P_G\bigl(-1/\epsilon_j^2\bigr) = 0 \,,
		\label{eq:spenergycond}
		\end{align} 
		where $P_G(x)$ is the vertex-weighted independence polynomial of the frustration graph $G(H)$,
		\begin{align}
		P_G(x) \coloneqq \sum_{k = 0}^{\alpha(G)} \sum_{S \in \mathcal{S}^{(k)}} \left(\prod_{\boldsymbol{j} \in S} b_{\boldsymbol{j}}^2\right) x^{k} \,.
		\label{eq:vertex_weighted_poly}
		\end{align}
		$\mathcal{S}^{(k)}$ is the set of $k$-vertex independent sets of $G(H)$, and $\alpha(G)$ is the independence number of $G(H)$. 
		\label{thm:maintheorem1}
	\end{theorem}
	
	An important result that we will show is that even if the frustration graph is only claw-free, then the Hamiltonian is still integrable, as there exists a set of mutually commuting conserved charges.
	We consider this a result of independent interest to many-body physicists.
	
	\begin{definition}[Independent-set charges]
		\label{def:independent_set_charge}
		Given a Hamiltonian of the form in Eq.~(\ref{eq:paulihdef}) with frustration graph $G(H)$, we define the $\alpha(G) + 1$ \emph{independent-set charges} as
		\begin{align}
		Q^{(k)} \coloneqq \sum_{S \in \mathcal{S}^{(k)}} \prod_{\boldsymbol{m} \in S} h_{\boldsymbol{m}}\,, \qquad k \in \{0, 1, \dots, \alpha(G) \}\,,
		\end{align}
		with the convention that $Q^{(0)} \coloneqq I$. Additionally, notice that $Q^{(1)} = H$. 
	\end{definition}
	
	As we will prove in \cref{lem:conserved_charges} below, the independent-set charges satisfy
	\begin{equation}	
	\big[Q^{(r)},Q^{(s)}\big]=0\,, \qquad \forall \ r,s \in \{1, \dots, \alpha(G) \}. 
	\end{equation}
	Since $Q^{(1)} = H$, this demonstrates that the charges are conserved.
	To take advantage of the independent-set charges, we exploit the simplicial property of $H$ and define a fiducial mode, $\chi$, in terms of which we can express the ``incognito modes". 
	
	\begin{definition}[Incognito mode, simplicial mode]
		\label{def:incognito}
		Given a simplicial Hamiltonian of the form in Eq.~(\ref{eq:paulihdef}) with frustration graph $G(H)$, we define a \textit{simplicial mode} $\chi$ to be any Pauli operator which is not present in the original Hamiltonian and which anticommutes with all of the operators in a simplicial clique of $G(H)$. 
		The $\alpha(G)$ \textit{incognito modes} of $H$ are defined with respect to a given simplicial mode $\chi$ as 
		\begin{align}
		\psi_j = N_j^{-1} T_G (-u_j) \chi T_G(u_j)\,, \qquad j \in \{1, \dots, \alpha(G) \} \,,
		\label{eq:incognitodef}
		\end{align}
		where $u_j \coloneqq 1/\epsilon_j$ for the single-particle energy $\epsilon_j$ satisfying \cref{eq:spenergycond}, $T_G(u)$ is a transfer operator
		\begin{align}
		T_{G} (u) \coloneqq \sum_{j = 0}^{\alpha(G)} (-u)^j Q^{(j)} \,,
		\label{eq:transfer_op}
		\end{align}
		and $N_j^{-1}$ is a normalization factor which is computable (see \cref{eq:normalization}).
	\end{definition}
	
	Note that we can always construct a simplicial mode for any simplicial Hamiltonian.
	To do so we introduce an additional (fictitious) spin to the system and augment each Hamiltonian term in the simplicial clique with a Pauli-X applied to the extra spin --- notice that this will not affect the frustration graph.
	The simplicial mode is then defined by a Pauli-Z operator applied to the additional spin; clearly, the simplicial mode will  anti-commute with all terms in the simplicial clique, but no other Hamiltonian terms.
	
	\begin{theorem}
		An ECF Hamiltonian $H$ is free-fermion-solvable via \cref{eq:ffsolvable} with eigenmodes given by its incognito modes.
		\label{thm:maintheorem2}
	\end{theorem}
	Our proofs of \Cref{thm:maintheorem1}~and~\Cref{thm:maintheorem2} closely resemble the solution method introduced by Fendley in Ref.~\cite{fendley2019free}. 
	Importantly, the fact that our Hamiltonians can be written in the form of \cref{eq:ffsolvable} implies that each energy level in the Hamiltonian has the same degeneracy, as similarly shown in  Ref.~\cite{fendley2019free}. 
	The operative technical insight is that many of the key properties of that model, and its generalizations in Refs.~\cite{Alcaraz2020Free, Alcaraz2020Integrable}, are actually special cases of more general recursion relations in the class of models we identify.
	
	\section{Relation to prior work}
	\label{sec:symmetries}
	Since its discovery, the Jordan-Wigner transformation~\cite{jordan1928uber} and subsequent generalizations~\cite{Fradkin1989JordanWigner, Wang1991Ground, Huerta1993BoseFermi, Shaofeng1995JordanWigner, Batista2001Generalized, Kitaev2006Anyons, Nussinov2009Bond, Galitski2010Fermionization, Cobanera2011Bond, backens2019jordan} have enjoyed great success in probing the fundamental physics of quantum many-body spin models, as well as classical statistical mechanics models through so-called transfer-matrix methods~\cite{Onsager1944Crystal, schultz1964twodimensional, kochmaski1997generalized}.
	An understanding of these mappings has furthermore proven useful for designing fermion-to-qubit mappings with desired properties for simulating fermionic systems on a quantum computer~\cite{Ortiz2001Quantum, Bravyi2002Fermionic, Verstraete2005Mapping, Nussinov2012Arbitrary, Bravyi2017Tapering, Havlek2017Operator, Steudtner2018Fermion, Setia2019Superfast, Jiang2019Majorana, Jiang2020Optimal}. 
	Here, operator locality in the dual spin model is generally enabled through coupling to an auxiliary gauge field~\cite{Chen2018Exact, Chen2019Bosonization, Tantivasadakarn2020JordanWigner}, which endows fermionic-pair excitations with the structure of freely deformable strings on the spin lattice~\cite{Levin2003Fermions}. 
	The preponderance of these mappings suggests that a fundamental theory of physics containing fermionic degrees of freedom need not hold fermions as fundamental objects~\cite{Ball2005Fermions, Wen2003Quantum}. 
	
	Free-fermion models have an interesting connection to graph theory. 
	The dynamics of free-fermion models are equivalent to matchgate circuits~\cite{Knill2001Fermionic,Terhal2002Classical, Bravyi2006Universal, Jozsa2008Matchgates, Brod2011Extending, Hebenstreit2019All}, which were originally developed in the context of counting perfect matchings in graphs~\cite{Valiant2002Quantum, Cai2006Valiants, Cai2007Theory, Valiant2008Holographic}. 
	Independently, graph-theoretic methods have been utilized in quantum information in the context of variational quantum eigensolvers~\cite{Crawford2019Efficient, Izmaylov2019Unitary, BonetMonroig2019Nearly, Gokhale2019Minimizing, Yen2020Measuring, Verteletskyi2020Measurement, Zhao2020Measurement}, where the frustration graph is commonly known as the \emph{anticompatibility} graph. 
	Inspired by these methods, two of the authors have shown that a generalized Jordan-Wigner transformation exists for exactly those models for which the frustration graph is a line graph~\cite{Chapman2020Characterization}. 
	
	\begin{definition}[Line Graph]
		A \emph{line graph} $L(R) \coloneqq (E, F)$ is a graph whose vertex set is in one-to-one correspondence with the edges $E$ of a \emph{root graph} $R \coloneqq (V, E)$. 
		Vertices in $L(R)$, $e_1$, $e_2 \in E$, are neighboring if and only if $|e_1 \cap e_2| = 1$, i.e. the edges in $R$ are incident at a vertex in $V$.
	\end{definition}
	
	We note that line graphs also play an important role in understanding the spectrum of certain tight-binding models~\cite{Kollar2019Line,Kollar2019Hyperbolic}, but we will not discuss these models further here. 
	
	A generalized Jordan-Wigner transformation maps a spin Hamiltonian of the form in Eq.~(\ref{eq:paulihdef}) to one which is quadratic in \emph{Majorana fermion} modes $\{\gamma_j\}$. 
	These are Hermitian operators, which satisfy canonical anticommutation relations
	\begin{align}
	\{\gamma_j, \gamma_k\} = 2 \delta_{jk} I \mbox{\hspace{5mm}} \mathrm{and} \mbox{\hspace{5mm}} \gamma_j^{\dagger} =\gamma_j \mbox{\hspace{5mm}} \forall \ j, \ k \,.
	\end{align}
	That is, when solving a Hamiltonian of the form in Eq.~(\ref{eq:paulihdef}) by Jordan-Wigner, we are asking whether there exists a mapping $\phi: V \mapsto \widetilde{V}^{\times 2}$ acting on the Pauli terms of $H$, and effecting
	\begin{align}
	\sigma^{\boldsymbol{j}} \mapsto i \gamma_{\phi_1(\boldsymbol{j})} \gamma_{\phi_2(\boldsymbol{j})} \mbox{\hspace{5mm}} \forall \ \boldsymbol{j} \in V
	\label{eq:phidef}
	\end{align}
	such that
	\begin{align}
	H \mapsto \widetilde{H} \coloneqq \frac{i}{2} \sum_{j, k \in \widetilde{V}} h_{jk} \gamma_j \gamma_k \coloneqq \frac{i}{2} \boldsymbol{\gamma}\cdot \mathbf{h} \cdot \boldsymbol{\gamma}^{\mathrm{T}}
	\label{eq:ffmap} \mathrm{,}
	\end{align}
	in a way that preserves the commutation relations between terms, i.e. $G(H) \simeq G(\widetilde{H})$. 
	The coefficient matrix $\mathbf{h}$ --- called the single-particle Hamiltonian --- is necessarily antisymmetric, since any symmetric part will vanish under the sum in Eq.~(\ref{eq:ffmap}), and we may take $H$ to be traceless without loss of generality. 
	The central theorem of Ref.~\cite{Chapman2020Characterization} gives a necessary and sufficient criterion for a generalized Jordan-Wigner solution to exist for a particular qubit Hamiltonian.
	
	\begin{theorem}[{Thm.\ 1 of Ref.~\cite{Chapman2020Characterization}}]
		An injective map $\phi$ as defined in \cref{eq:phidef} and \cref{eq:ffmap}
		such that $G(H) \simeq G(\widetilde{H})$ exists for the Hamiltonian $H$ as defined in \cref{eq:paulihdef} if and only if there exists a root graph $R$ such that
		\begin{align}
		G(H) \simeq L(R) \,.
		\label{eq:lineequiv}
		\end{align}
		\label{eq:thm1chapman}
	\end{theorem}
	Upon constructing a free-fermion solution for a given Hamiltonian, we find that $\mathbf{h}$ gives an edge-weighted skew-adjacency matrix of the root graph $R$.
	The graph $R$ may therefore be seen as the Majorana-fermion hopping graph. 
	A full solution for $H$ is found by a linear transformation on the Majorana modes in Eq.~(\ref{eq:ffmap}) to diagonalize the Hermitian matrix $i\mathbf{h}$. 
	Letting the nonzero eigenvalues of $i\mathbf{h}$ be given by $\{\pm \epsilon_j\}_{j = 1}^{\alpha}$ with $\epsilon_j > 0$ for all $j \in \{1, \dots, \alpha\}$, this brings the Hamiltonian to the form in Eq.~(\ref{eq:ffsolvable}) with single-particle energies given by the $\{\epsilon_j\}$.
	
	\begin{table*}
		\centering
		\setcellgapes{3pt}
		\makegapedcells
		\begin{tabular} {c  c  c}
			\toprule
			\multicolumn{2}{c}{With Twins}& \\ \cmidrule{1-2}
			\multicolumn{1}{c}{}& \multicolumn{2}{c}{(Even-hole, Claw)-Free} \\ \cmidrule{2-3}
			\makecell{\includegraphics[width=0.14 \textwidth]{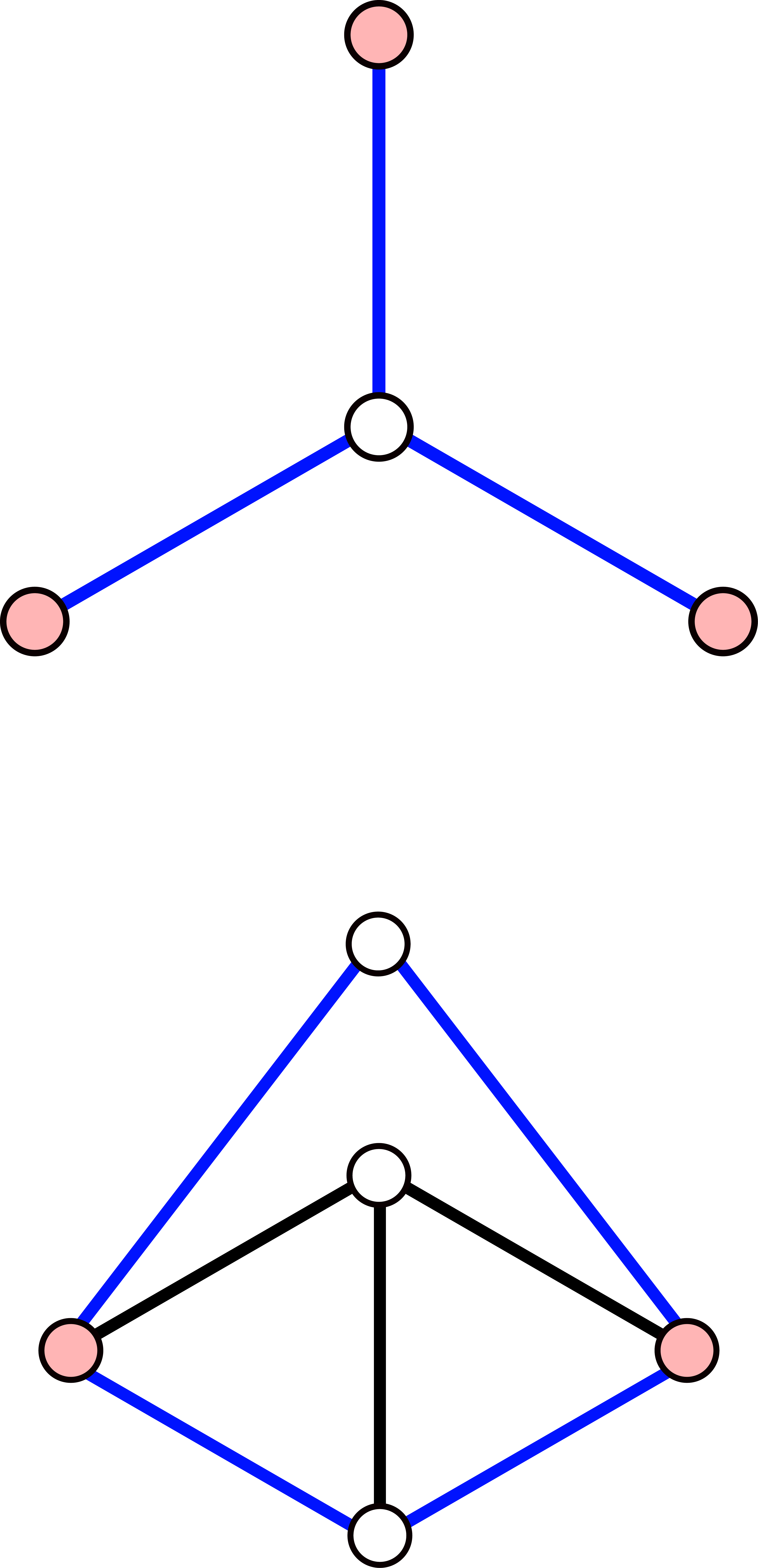}} & 	\makecell{\includegraphics[width=0.14 \textwidth]{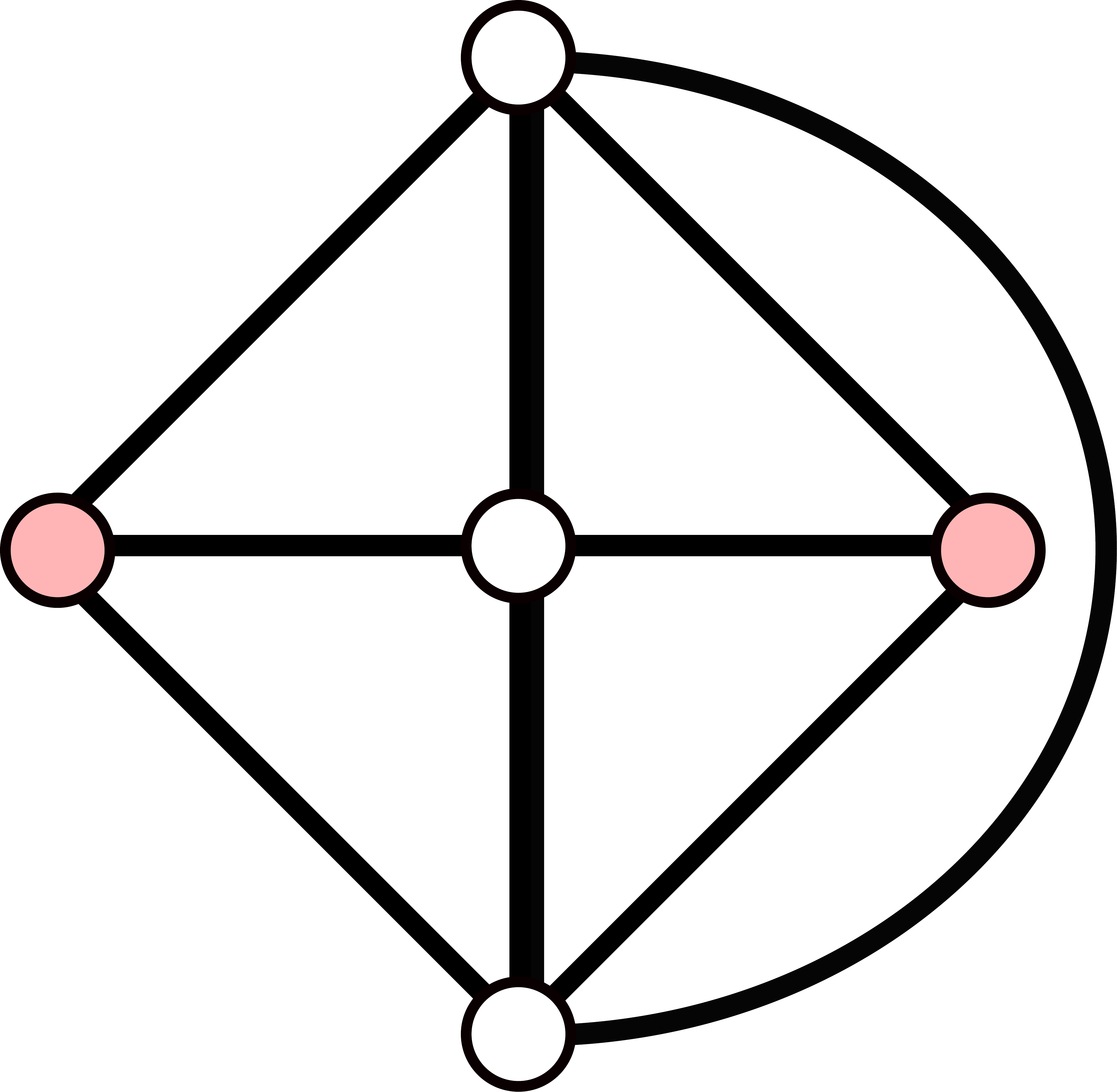}} &
			\makecell{\includegraphics[width=0.62 \textwidth]{LineForbiddenwoTwinsorEvenHoleorClaw}} \\
			\bottomrule
		\end{tabular}
		\caption{
			Nine forbidden induced subgraphs for line graphs. 
			No model containing a subset of terms inducing any of the above frustration graphs admits a generalized Jordan-Wigner solution, unless the global frustration graph contains twin vertices which may then be removed via a symmetry to give a line graph~\cite{Chapman2020Characterization}. 
			Note that we define twin vertices to be non-neighboring, and that all but the claw and the graphs in the rightmost column above contain at least one pair of \emph{neighboring} vertices with the same closed neighborhoods. 
			Vertices from these pairs may be removed through a unitary rotation as described in \Cref{sec:smallsystems}.
			The six-qubit instance of the four-fermion model introduced by Fendley has the frustration graph shown at the very bottom right, and larger instances clearly contain a subset of terms inducing this graph. 
			Surprisingly, each of these graphs describes a standalone model with a free-fermion solution, as the graphs themselves either contain twins or are (even-hole, claw)-free. 
			What is important, therefore, is the precise way that these graphs are connected to a global frustration graph that determines whether or not a generic free-fermion solution is possible. 
			\textbf{Left column:} Forbidden subgraphs which contain twin vertices, highlighted in red, but also contain either a claw or an even-hole, highlighted by blue edges. 
			\textbf{Middle column:} This forbidden subgraph contains twin vertices, but no claws or even holes. Each red highlighted vertex is also simplicial.
			\textbf{Right column:} These graphs do not contain twins, but are (even-hole, claw)-free. 
			Though they contain many simplicial cliques, an example is highlighted for each graph in orange. 
			Importantly, note that the simplicial vertex highlighted in the four-fermion model is necessary for Fendley's exact solution of this model.
		}
		\label{tab:forbidden_lines}
	\end{table*}
	
	Claw-free graphs were originally investigated as natural generalizations to line graphs~\cite{chudnovsky2005structure} and have since developed into the subject of a rich area of study in graph theory \cite{faudree1997clawfree}. 
	Line graphs are claw-free, as no three edges in $R$ can be incident to another edge without at least two of them being incident to each other. 
	Remarkably, the free-fermion solution method presented in this work extends the generalized Jordan-Wigner solution in a way that parallels the relationship between claw-free graphs and line graphs. 
	Specifically, the single-particle energies obtained from Eq.~(\ref{eq:ffmap}) satisfy Eq.~(\ref{eq:spenergycond}) when $L(R)$ is an ECF graph. Every vertex of $R$ corresponds to a simplicial clique of $L(R)$. 
	Clearly, the edges incident to a given vertex in $R$ are mapped to the vertices of a clique in $L(R)$, and since every edge is incident to two vertices in $R$, the open neighborhood of a vertex in $L(R)$ induces two vertex-disjoint cliques ~\cite{krausz1943demonstration}. 
	When a line graph is even-hole-free, the corresponding root graph is even-cycle-free, as any even cycle in $R$ will be mapped to an even hole in $L(R)$.
	Suppose that a given Hamiltonian satisfies Eq.~(\ref{eq:lineequiv}) with $L(R)$ an ECF graph. 
	The single particle energies are zeros of 
	\begin{align}
	f_R(u) = \det\left(i\mathbf{h} - u \mathbf{I}\right) \,.
	\end{align}
	Equivalently, we may consider the reciprocal polynomial $f^*_R$ to $f_R$,
	\begin{align}
	f^*_R(u) \coloneqq u^n f_R(1/u) = (-1)^n \det\left(\mathbf{I} - i u \mathbf{h}\right) \,,
	\end{align}
	where $n$ is the number of vertices in $R$. 
	For an arbitrary simple graph $R$, the characteristic polynomial $f_R$ (and thus $f_R^*$) would only depend on products of elements from $\mathbf{h}$ from matchings and even cycles of $R$ \cite{cavers2012skew}. 
	Since $R$ is an even-cycle-free graph however, only the matchings are relevant. 
	Let $\mathcal{M}_k$ be the set of all $k$-edge matchings $M$ of $R$. 
	We have
	\begin{align}
	f^*_R(u) &= (-1)^n \sum_{k = 0}^{\lfloor n/2 \rfloor} (-u^2)^k \left[\sum_{M \in \mathcal{M}_k} \prod_{(i, j) \in M} h_{ij}^2\right] \label{eq:matchingpoly}\\ 
	f^*_R(u) &= (-1)^n P_{L(R)}(-u^2) \,.
	\end{align}
	The last equality follows because the matchings of a graph correspond to the independent sets of its line graph. 
	Therefore, $\pm \epsilon_j$ are an eigenvalue pair of $i \mathbf{h}$ if and only if Eq.~(\ref{eq:spenergycond}) is satisfied for $G(H) \simeq L(R)$. 
	Though even-hole-free line graphs are a rather limited set of frustration graphs, what is incredibly surprising is that Theorem \ref{thm:maintheorem1} holds when $G(H)$ is relaxed to be a general ECF graph, though there is no fermion-hopping graph $R$ for this set of graphs in general. 
	We remark that for any claw-free graph, the vertex-weighted independence polynomial, $P_G(x)$, is real-rooted for all values of the Hamiltonian couplings by the results given in Refs.~\cite{Leake2016Generalizations, Engstrom07inequalitieson}. 
	This generalizes the result for the un-weighted independence polynomial originally proved by Chudnovsky and Seymour~\cite{Chudnovsky2007Roots}. 
	Since $P_G(x)$ has non-negative coefficients and real roots, all of its roots must be negative by Descartes' rule of signs \cite{descartesgeometrie}, and therefore all of the single-particle energies $\{\epsilon_j\}$ are themselves real. 
	The free-fermion solution method presented here therefore includes systems for which we can prove that no generalized Jordan-Wigner solution is possible, as we shall now see.
	
	Line graphs can be characterized by the set of nine forbidden subgraphs~\cite{beineke1970characterizations}, shown in \cref{tab:forbidden_lines}.
	No model containing a subset of terms inducing any of the frustration graphs in \cref{tab:forbidden_lines} admits a generalized Jordan-Wigner solution. 
	One possible exception is when the Jordan-Wigner mapping is allowed to be non-injective; i.e.~there is a mapping satisfying Eq.~(\ref{eq:phidef}) and preserving the frustration graph which takes multiple Pauli terms to the same fermionic pair. 
	These Pauli terms must then correspond to \emph{twin vertices} in $G(H)$: vertices with identical open neighborhoods, $\Gamma(v)$. 
	Notice that twin vertices are never neighboring by this definition, as a vertex is not included in its own open neighborhood.%
	\footnote{Here we caution the reader that this definition differs slightly from that used in the graph-theory literature, where vertices with identical closed neighborhoods (which are therefore neighboring) are also referred to as twins. 
		We will return to pairs of vertices with identical closed neighborhoods in \Cref{sec:smallsystems}.}
	Since operators corresponding to twin vertices anticommute with the same set of operators in the Hamiltonian, the product of any pair of such operators commutes with every term in the Hamiltonian and so constitutes a symmetry. 
	We can thus project onto the eigenspace of this symmetry operator to replace one operator in the set of twins with another, thus removing its vertex from the frustration graph. 
	If twins can be removed in such a way as to change the frustration graph into a line graph, then the Hamiltonian is still solvable via Jordan-Wigner. 
	This will sometimes be possible, as some of the forbidden subgraphs in \cref{tab:forbidden_lines} themselves contain twins. 
	From \cref{tab:forbidden_lines}, we see that all forbidden subgraphs for line graphs either contain twins, are simplicial ECF, or both. 
	They therefore surprisingly all have a free spectrum, and it is truly how these graphs are connected to one another that allows us to infer the existence of a free-fermion solution. 
	
	The class of ECF graphs is generalized by the set of so-called (even-hole, pan)-free graphs~\cite{Cameron2017On}. 
	A \emph{pan} is a graph consisting of a hole together with an additional vertex with exactly one neighbor on the hole. 
	A pan contains a claw as an induced subgraph, and so an (even-hole, pan)-free graph is necessarily ECF.
	The structure of (even-hole, pan)-free graphs has been completely characterized, and this allows the authors of Ref.~\cite{Cameron2017On} to give an $O(mn)$-time algorithm for recognizing them, where $m$ is the number of edges in the graph and $n$ is the number of vertices. 
	Specifically, the authors of Ref.~\cite{Cameron2017On} show that (even-hole, pan)-free graphs either: (i) have a clique cutset, (ii) are unit circular-arc graphs, (iii) are a clique, (iv) are the join of a clique and a unit circular-arc graphs. 
	A unit circular-arc graph is one whose vertices correspond to distinct arcs of unit length on a circle, such that vertices are neighboring if and only if their corresponding arcs intersect. 
	A clique cutset is a subset of vertices inducing a clique whose removal disconnects the graph. 
	The join of two graphs $G_1 \coloneqq (V_1, E_1)$ and $G_2 \coloneqq (V_2, E_2)$ is the graph with vertex-set $V_1 \cup V_2$ and edge-set $E_1 \cup E_2 \cup \{(u, v)| u \in V_1, v \in V_2\}$. 
	Though the theorem of Ref.~\cite{Cameron2017On} completely describes the structure of these graphs, we can intuitively expect that the coarse topology of these models is fundamentally one-dimensional or treelike. 
	
	As general claw-free graphs can also be recognized efficiently~\cite{Kloks2000Finding}, we have an efficient algorithm for recognizing ECF graphs. 
	In Ref.~\cite{Chudnovsky2012Growing}, a polynomial time algorithm is given for detecting whether a general claw-free graph contains a simplicial clique.
	It is therefore computationally efficient to recognize a simplicial clique in an ECF graph. 
	Moreover, every (nonempty) ECF graph has at least one simplicial clique~\cite{CSSS}. 
	
	\section{Proofs of Main Results}
	\label{sec:technical}
	In this section we prove the two main results presented in~\Cref{sec:informal}. 
	We restate these theorems here for convenience.
	The first main result tells us that an ECF model will have a free spectrum, of the form \cref{eq:free_energy}, and provides an explicit form of the single-particle energies. 
	\begin{customthm}{1}
		\label{thm:theorem1}
		{\upshape \textbf{(Restatement.)}} Every ECF Hamiltonian $H$ has a free spectrum of the form in \cref{eq:free_energy}. 
		In particular, the single-particle energies $\{\epsilon_j\}_{j=1}^{\alpha(G)}$ satisfy
		\begin{align}
		P_G\bigl(-1/\epsilon_j^2\bigr) = 0 \,,
		\label{eq:spenergycond2}
		\end{align} 
		where $P_G(x)$ is the vertex-weighted independence polynomial of the frustration graph $G(H)$,
		\begin{align}
		P_G(x) \coloneqq \sum_{k = 0}^{\alpha(G)} \sum_{S \in \mathcal{S}^{(k)}} \left(\prod_{\boldsymbol{j} \in S} b_{\boldsymbol{j}}^2\right) x^{k} \,.
		\end{align}
		$\mathcal{S}^{(k)}$ is the set of $k$-vertex independent sets of $G(H)$, and $\alpha(G)$ is the independence number of $G(H)$.
	\end{customthm}
	
	The second main result gives an explicit realization of the canonical modes of an ECF model in terms of independent sets of Hamiltonian terms, and the simplicial mode, $\chi$.
	\begin{customthm}{2}
		\label{thm:theorem2}
		{\upshape \textbf{(Restatement.)}} An ECF Hamiltonian $H$ is free-fermion-solvable via \cref{eq:ffsolvable} with eigenmodes given by its incognito modes.
	\end{customthm}
	Recall that the incognito modes are defined in terms of the simplicial mode, $\chi$, as
	\begin{align}
	\psi_j = N_j^{-1} T_G (-u_j) \chi T_G(u_j)\,, \qquad j \in \{1, \dots, \alpha(G) \} \,,
	\end{align}
	where $u_j \coloneqq 1/\epsilon_j$ for the single-particle energy $\epsilon_j$ satisfying \cref{eq:spenergycond2}, $T_G(u)$ is a transfer operator
	\begin{align}
	T_{G} (u) \coloneqq \sum_{j = 0}^{\alpha(G)} (-u)^j Q^{(j)} \,,
	\label{eq:transfer_op_restate}
	\end{align}
	and $N_j^{-1}$ is a normalization factor which is computable.
	
	\begin{figure}[t!]
		\centering
		\includegraphics[width=0.9\textwidth]{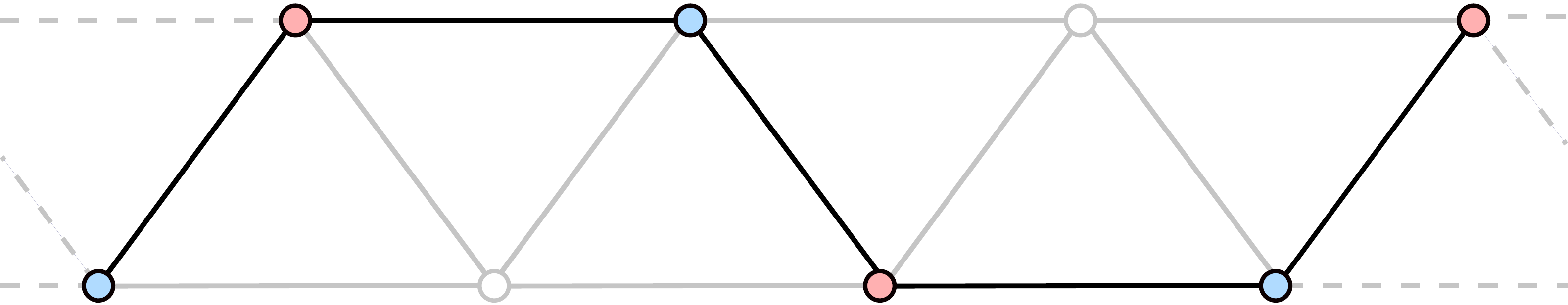}
		\caption{The frustration graph of the model introduced in Ref.~\protect\cite{fendley2019free}, which is (even-hole, claw)-free. Each of the sets of colored (red and blue) vertices are independent sets, and together they induce a path in the frustration graph. In a general claw-free graph, the symmetric difference of any pair of independent sets induces a bipartite subgraph of maximum degree two: all of the components of the subgraph are induced paths and even-holes. If the graph is furthermore even-hole free, then the symmetric difference induces a set of disjoint paths.}
		\label{fig:fourfermionindset}
	\end{figure}
	
	We proceed by making successively more restrictive assumptions on $G(H)$: first that it is claw-free, then (even-hole, claw)-free. 
	We begin by proving the following lemma, regarding claw-free Hamiltonians and the independent-set charges.
	\begin{lemma}
		\label{lem:conserved_charges}
		Given a Hamiltonian with claw-free frustration graph $G(H)$, the independent-set charges are mutually commuting:
		\begin{equation}	\big[Q^{(r)},Q^{(s)}\big]=0 \,, \qquad \forall \ r,s \in \{1, \dots, \alpha(G) \}.
		\label{eq:chargecomm}
		\end{equation}
	\end{lemma}
	
	\begin{proof}
		We may assume $r \neq s$, since \cref{eq:chargecomm} clearly holds if $r$ and $s$ are equal. For a given independent set $S$, define
		\begin{align}
		h_S \coloneqq \prod_{\boldsymbol{j} \in S} h_{\boldsymbol{j}} \,.
		\end{align}
		and notice that, since operators belonging to an independent set in $G(H)$ are commuting, the order in which we take the product is unimportant in this definition. For two independent sets $S$, $S'$ of a claw-free graph
		\begin{align}
		[h_S, h_{S'}] =
		\begin{cases}
		\pm 2 \left(\prod_{\boldsymbol{j} \in S \cap S'} b_{\boldsymbol{j}}^2 \right) \prod_{\boldsymbol{j} \in S \oplus S'} h_{\boldsymbol{j}} & |E_{S \oplus S'}| \ \mathrm{odd} \\
		0 & |E_{S \oplus S'}| \ \mathrm{even} 
		\end{cases}
		\label{eq:hshs}
		\end{align}
		where $S \oplus S' \coloneqq (S \cup S')\backslash (S \cap S')$, the symmetric difference of $S$ with $S'$.
		When it is not empty, the graph $G[S \oplus S']$ is bipartite, since $S$ and $S'$ are both independent sets.
		Commuting $h_S$ through $h_{S'}$ thus gives a factor of $-1$ for every edge in this graph, and so Eq.~(\ref{eq:hshs}) holds. From here, we naturally restrict to the case where $|E_{S \oplus S'}|$ is odd.
		
		As $G$ is claw-free, we must furthermore have that every vertex of $G[S \oplus S']$ has degree at most two in this graph, since once again, $G[S \oplus S']$ is bipartite.
		Every component of $G[S \oplus S']$ is therefore either an isolated vertex, path, or even cycle (odd cycles are not bipartite).
		We have assumed that $G[S \oplus S']$ has odd-many edges, and so this graph must have an odd number (and thus at least one) of odd-length-path components.
		Such paths have the same number of vertices from both $S$ and $S'$ and so cannot be the only component of $G[S \oplus S']$, since we have assumed $r \neq s$.
		Pick one such odd path, $L \subseteq V$, and note that
		\begin{align}
		\{h_{S \cap L},h_{S' \cap L}\} = 0.
		\end{align} 
		Since $G[L]$ has the same number of vertices from both $S$ and $S'$, we can exchange the subsets $S \cap L$ and $S' \cap L$ between $S$ and $S'$, respectively, to obtain a new unique pair of independent sets without changing the number of vertices in either, while also preserving the sets $S \cap S'$ and $S \oplus S'$. 
		This gives
		\begin{align}
		& [h_{S/L}h_{S \cap L}, h_{S'/L} h_{S' \cap L}] = -[h_{S/L} h_{S' \cap L}, h_{S'/L} h_{S \cap L}] \,,
		\end{align}
		and so these terms cancel in the commutator $[Q^{(r)}, Q^{(s)}]$. Letting $N$ be the number of odd-length-path components in $G[S \oplus S']$, there are are $2^{N}$ pairs of independent sets $(S, S')$, related by these exchanges, for which the graph $G[S \oplus S']$ is fixed. The contributions to the commutator $[Q^{(r)}, Q^{(s)}]$ from each $(S, S')$ therefore cancel pairwise,  and we have $[Q^{(r)}, Q^{(s)}] = 0$ for all $r$ and $s$.\qed
	\end{proof}

	\Cref{lem:conserved_charges} implies that all claw-free models have a set of conserved quantities whose size generally grows with system size, since $Q^{(1)}\coloneqq H$.
	Thus, we can conclude that, in the traditional sense, claw-free models are \textit{integrable}.
	Consider as an example the frustration graph of the model introduced in Ref.~\protect\cite{fendley2019free}, as shown in \cref{fig:fourfermionindset}. The graph is always claw-free, although contains even holes when the model has periodic boundary conditions.
	Two independent sets of vertices (highlighted in red and blue) induce a path in the frustration graph. In a general claw-free graph, the symmetric difference of any pair of independent sets induces a bipartite subgraph of maximum degree two: all of the components of the subgraph are induced paths and even-holes.
	\cref{lem:conserved_charges} also implies that the transfer operator, $T_G(u)$, defined in \cref{eq:transfer_op_restate}, will commute with the Hamiltonian.

	Next consider the following lemma regarding the transfer operators, $T_G(u)$, of even-hole-free models.
	
	\begin{lemma}
		\label{lem:transfer_poly}
		If $G$ is an (even-hole, claw)-free graph, the transfer matrix, $T_G(u)$, satisfies
		\begin{equation}
		T_G(u)T_G(-u)=P_G(-u^2)
		\label{eq:lemma2restate}
		\end{equation}
		where $P_G$ is the vertex-weighted independence polynomial, defined in~\cref{eq:vertex_weighted_poly}.
	\end{lemma}
	
	\begin{proof}
		Let $G$ be an ECF graph. Using \cref{eq:transfer_op_restate} we have
		\begin{align}
		T_G T_G^- &= \sum_{s, t = 0}^{\alpha(G)} (-1)^s u^{s + t} Q^{(s)} Q^{(t)},
		\end{align}
		where we have used the abbreviated notation  $T_G(-u)\coloneqq T_G^-$.
		If $s$ and $t$ have opposite parity, then $Q^{(s)}Q^{(t)}$ and $Q^{(t)}Q^{(s)}$ have a relative minus sign in the sum, and so these terms vanish in the sum since $Q^{(s)}$ and $Q^{(t)}$ commute. 
		
		Thus we need only consider terms for which $s$ and $t$ have the same parity
		\begin{align}
		T_G T_G^- &= \sum_{\substack{s, t = 0 \\ s + t \ \mathrm{even}}}^{\alpha(G)} (-1)^s u^{s + t} Q^{(s)}Q^{(t)},
		\end{align}
		By expanding the $Q^{(k)}$ in terms of independent sets, $h_S$, we can write
		\begin{align}
		T_G T_G^- &= \sum_{\substack{s, t = 0 \\ s + t \ \mathrm{even}}}^{\alpha(G)} (-1)^s u^{s + t}\sum_{\substack{S \in \mathcal{S}^{(s)} \\ S' \in \mathcal{S}^{(t)} \\ |E_{S \oplus S'}| \ \mathrm{even}}} \left(\prod_{\boldsymbol{j} \in S \cap S'} b_{\boldsymbol{j}}^2\right)h_{S \cap (S \oplus S')} h_{S' \cap (S \oplus S')}.
		\end{align}
		The constraint that $s + t$ is even implies that the number of vertices $|V_{S \oplus S'}|$ is even, and we require that $|E_{S \oplus S'}|$ be even because the operators $h_S$ and $h_{S'}$ will anticommute otherwise and cancel in the sum over $S$, $S'$. It thus suffices to consider induced subgraphs, $G[S \oplus S']$, with even-many edges and even-many vertices.
		Once again, such graphs must be bipartite and, furthermore, must be a union of disjoint isolated vertices, paths, and even cycles. 
		
		By a similar argument as above, we will show that the contributions from any such graphs containing an odd-length path will cancel in the sum.
		Assume that $G[S \oplus S']$ does contain an odd-length path $L$.
		Since $|E_{S \oplus S'}|$ must be even, $L$ cannot be the only component of $G[S \oplus S']$, and in fact one of the additional components of $G[S \oplus S']$ must be another odd-length path (otherwise the total number of edges in $G[S \oplus S']$ cannot be made even).
		Exchanging $S \cap L$ and $S' \cap L$ between $S$ and $S'$ gives another pair of distinct independent sets for the same $s$, $t$, $S \cap S'$, and $S \oplus S'$, but for which the operator $h_{S \cap (S \oplus S')} h_{S' \cap (S \oplus S')}$ appears with a minus sign in the sum and cancels the term corresponding to $S$ and $S'$.
		The contributions from $G[S \oplus S']$ therefore cancel pairwise in this case.
		
		Next, we will show that contributions from any such graphs containing an even-length path will cancel in the sum, and therefore non-vanishing contributions must come from graphs containing no paths at all.
		Once again, assume $G[S \oplus S']$ does contain such a path $L$ of even length (which may be an isolated vertex, i.e. a path of length zero).
		Since $L$ has an odd number of vertices, $L$ cannot be the only component of $G[S \oplus S']$ and in fact one of the additional components of $G[S \oplus S']$ must be another even-length path (otherwise the total number of vertices in $G[S \oplus S']$ cannot be made even).
		Both of the endpoints of $L$ must belong to the same independent set, either $S$ or $S'$.
		If $L$ is an isolated vertex, then it trivially belongs to the same independent set as itself.
		Exchanging $S \cap L$ and $S' \cap L$ between $S$ and $S'$ in this case gives another pair of distinct independent sets for the same value of $s + t$, $S \cap S'$, and $S \oplus S'$, for which the operator $h_{S \cap (S \oplus S')} h_{S' \cap (S \oplus S')}$ appears with the same sign in the sum since $L$ has even-many edges. Both of the parities of $s$ and $t$ are changed in this exchange, and so $s$ and $t$ have the same parity still, but this term appears with an overall relative minus sign in the sum due to the factor of $(-1)^s$.
		This therefore cancels the term corresponding to $S$ and $S'$, and so the contributions from $G[S \oplus S']$ cancel pairwise.
		
		The only allowed graphs $G[S \oplus S']$ whose term in the sum is not canceled by something else are those for which $G[S \oplus S']$ consists of a set of disconnected even cycles. However, we have assumed that $G$ is even-hole free. Therefore, these contributions do not appear, and we will have
		\begin{align}
		T_GT_G^- = P_G(-u^2),
		\label{eq:ehfidentity}
		\end{align}
		if there are no even holes in $G$.
		\qed	
	\end{proof}

	Note that $P_G$ has strictly positive coefficients, which do not depend on the signs of the Hamiltonian coefficients $\{b_{\boldsymbol{j}}\}_{\boldsymbol{j} \in V}$.
	Thus, as discussed in \cref{sec:symmetries}, $P_G(-x)$ will have all positive roots, denoted by $x\coloneqq u_{\ell}^2$.
	
	We next consider the commutation of the incognito modes, $\psi_\ell$, with the Hamiltonian. 
	Here we further use the fact that an ECF graph, $G(H)$, contains a simplicial clique, $K_s$.
	Recall that the simplicial mode, $\chi$, commutes with all terms in the Hamiltonian outside of $K_s$, but anticommutes with all terms in $K_s$, $\{\chi,h_v\}=0$ for all $v\in K_s$.
	Thus, we can write the commutation of $T_G\chi T_G^-$ with the Hamiltonian for arbitrary $u$ as
	\begin{equation}
	[H,T_G(u)\chi T_G(-u)]=2\sum_{v\in K_s}T_G(u)h_v\chi T_G(-u).
	\label{eq:Hcommchi}
	\end{equation}
	For an ECH model, we require that when $u=-u_\ell$, the right-hand-side of \cref{eq:Hcommchi} is equal to~$2\epsilon_\ell\psi_{\ell}$, where $1/\epsilon_\ell\coloneqq u_\ell$ (similarly when $u=u_\ell$, it is equal to~$-2\epsilon_\ell\psid_{\ell}$). 
	A crucial step for proving this is the following lemma:
	
	\begin{lemma}
		\label{lem:simplicialTcomm}
		Let $K_s$ be a simplicial clique in $G(H)$, and let $\chi$ be a simplicial mode, as defined in~\cref{def:incognito}, then
		\begin{equation}
		T_G \left(1 + u \sum_{v\in K_s}h_v\right) \chi T^-_G = P_G(-u^2) \left(1 - u \sum_{v\in K_s}h_v\right) \chi.
		\label{eq:fundidentity}
		\end{equation}
	\end{lemma}
	\begin{proof}
		We first express important recurrence relations for both $T_G$ and $P_G$.
		For any clique $K\subseteq G$ we have
		\begin{align}
		T_G =& T_{G - K} - u \sum_{v \in K} h_v T_{G-\Gamma[v]}. 
		\label{eq:TKrecurrence}
		\end{align}
		This follows from the fact that independent sets of $G$ can be partitioned into two groups: (i) sets which do not contain $v\in K$, corresponding to the first term $T_{G-K}$; and (ii) sets which contain a single $v\in K$, and thus contain none of its neighbors, corresponding to the second term, $-u\sum_{v\in K}h_vT_{G-\Gamma[v]}$.
		When $K$ is simplicial, $K\coloneqq K_s$, we can show the additional recursion relation
		\begin{align}
		T_G= T_{G - K_s} - u \sum_{v \in K_s} h_v T_{G-K_v}
		\label{eq:TKrecurrence2}
		\end{align}
		where $K_v \coloneqq \Gamma[v] \backslash (K_s \backslash v)$ is a clique in $G$ for all $v \in K_s$, since $K_s$ is simplicial. 
		We show Eq.~(\ref{eq:TKrecurrence2}) by applying the recursion relation in Eq.~(\ref{eq:TKrecurrence}) twice in succession
		\begin{align}
		T_G &= T_{G - K_s} - u \sum_{v \in K_s} h_v T_{G- K_s - K_v} \label{eq:TKrecurrence2line1}\\
		T_G &= T_{G - K_s} - u \sum_{v \in K_s} h_v \left(T_{G - K_v} + u \sum_{w \in K_s\backslash\{v\}} h_w T_{G - K_s - K_v - K_w}\right) \label{eq:TKrecurrence2line2}
		\end{align}
		where we have rearranged the expansion Eq.~(\ref{eq:TKrecurrence}) by the clique $K_s$ in the graph $G - K_v$ and substituted into Eq.~(\ref{eq:TKrecurrence2line1}) to obtain Eq.~(\ref{eq:TKrecurrence2line2}) (recall that, by definition, $K_v \cap K_s = \{v\}$). 
		Expanding Eq.~(\ref{eq:TKrecurrence2line2}), we see that the operators $h_v$ and  $h_w$ anticommute since $v$ and $w$ are distinct vertices both belonging to the clique $K_s$. 
		However, the subscript of the transfer matrix is symmetric under the exchange of $v$ and $w$ in the double sum. 
		This double sum over $v \neq w \in K_s$ therefore vanishes and we obtain the desired relation in Eq.~(\ref{eq:TKrecurrence2}).
		
		Notice that both \cref{eq:TKrecurrence,eq:TKrecurrence2} have analogues in terms of $T_G^-$, given by substituting $u$ for $-u$ in these identities. Additionally, both \cref{eq:TKrecurrence,eq:TKrecurrence2} have analogues with $h_v$ to the right of the transfer operator instead of to the left. It is especially surprising that this is true for \cref{eq:TKrecurrence2}, since $h_v$ does not commute with $T_{G - K_v}$ in general. Examining the proof however, we see that we can equivalently pull $h_v$ to the right instead of to the left everywhere, and the proof goes through. In the forthcoming proofs, we will often refer to our use of these analogous identities as \cref{eq:TKrecurrence,eq:TKrecurrence2}, as the specific form of the identity we are using will be clear from context.

		By similar reasoning as for $T_G$, we have the corresponding recurrence relation for $P_G(-u^2)$
		\begin{align}
		P_G =& P_{G - K} - u^2 \sum_{v \in K} b^2_v P_{G-\Gamma[v]}
		\label{eq:PKrecurrence}
		\end{align}
		Note that, since any induced subgraph of $G$ is also ECF, we can expand Eq.~(\ref{eq:lemma2restate}) in Lemma  \ref{lem:transfer_poly} by Eq.~(\ref{eq:TKrecurrence}) for some clique $K$ to obtain 
		\begin{align}
		P_G(-u^2) &= T_GT_G^- \\
		&= \left(T_{G - K} - u \sum_{v \in K} h_v T_{G-\Gamma[v]}\right) \left(T^-_{G - K} + u \sum_{v \in K} h_v T^-_{G-\Gamma[v]}\right) \\
		&= P_{G - K} - u^2 \sum_{v \in K} b_v^2 P_{G - \Gamma[v]} + u \sum_{v \in K} \left(T_{G - K} h_v T^-_{G - \Gamma[v]} - h_v T_{G - \Gamma[v]} T_{G - K}^-\right) \nonumber \\
		& \mbox{\hspace{5mm}} -u^2 \sum_{v \neq w \in K} h_v T_{G - \Gamma[v]} h_w T^-_{G - \Gamma[w]} \\
		P_G(-u^2) &= P_G(-u^2) + u \sum_{v \in K} \left(T_{G - K} h_v T^-_{G - \Gamma[v]} - h_v T_{G - \Gamma[v]} T_{G - K}^-\right) \nonumber \\ 
		& \mbox{\hspace{5mm}} -u^2 \sum_{v \neq w \in K} h_v T_{G - \Gamma[v]} h_w T_{G - \Gamma[w]}^-
		\end{align}
		In the last line, we used the recurrence relation in Eq.~(\ref{eq:PKrecurrence}). 
		This gives
		\begin{align}
		u^2 \sum_{v \neq w \in K} h_v T_{G - \Gamma[v]} h_w T_{G - \Gamma[w]}^- = u \sum_{v \in K} \left(T_{G - K} h_v T^-_{G - \Gamma[v]} - h_v T_{G - \Gamma[v]} T_{G - K}^-\right) \,.
		\label{eq:clique-decomp}
		\end{align}
		
		Now we expand the left-hand side of \Cref{eq:fundidentity}, and compute each of the two terms:
		\begin{equation}
		T_G \left(1 + u \sum_{v\in K_s}h_v\right) \chi T^-_G =T_G\chi T^-_G + u \sum_{v\in K_s}T_Gh_v\chi T_G^-
		\label{eq:expand_left}
		\end{equation}
		For the first term, we make use of the recurrence relation \cref{eq:TKrecurrence} for the simplicial clique, $K_s$, noting that $\chi$ anticommutes with $h_v$ for all $v\in K_s$ and commutes with $h_v$ for $v \notin K_s$
		\begin{align}
		T_G\chi T_G^-&= \left(T_{G - K_s} - u \sum_{v \in K_s} h_v T_{G - \Gamma[v]} \right) \left(T^-_{G - K_s} - u \sum_{v \in K_s} h_v T^-_{G - \Gamma[v]}\right) \chi \\
		&= \left[P_{G - K_s} + u^2 \sum_{v \in K_s} b_v^2 P_{G - \Gamma[v]} - u \sum_{v \in K_s} \left(T_{G - K_s} h_v T^-_{G - \Gamma[v]} + h_v T_{G - \Gamma[v]} T_{G - K_s}^-\right) \right.\nonumber \\
		& \mbox{\hspace{5mm}} \left.+u^2 \sum_{v \neq w \in K_s} h_v T_{G - \Gamma[v]} h_w T^-_{G - \Gamma[w]}\right] \chi \\
		T_G\chi T_G^- &= \left(P_G + 2u^2 \sum_{v \in K_s} b_v^2 P_{G - \Gamma[v]} - 2u\sum_{v\in K_s}h_v T_{G - \Gamma[v]} T^-_{G-K_s}
		\right)\chi \label{eq:chitermline4}
		\end{align}
		In the last line, we used the recurrence relation Eq.~(\ref{eq:PKrecurrence}) and the identity Eq.~(\ref{eq:clique-decomp}).
		
		Turning to the second term in \cref{eq:expand_left}, we consider the individual terms in the sum separately. 
		For each $v \in K_s$, we expand by $K_v$ using Eq.~(\ref{eq:TKrecurrence}). 
		We then use the fact that $h_v \chi$ anticommutes with $h_w$ for all $w \in K_v$ and commutes with $h_w$ for all $w \notin K_v$. 
		A similar set of steps as above gives
		\begin{equation}
		T_Gh_v\chi T_G^-=\Bigg(2P_{G-K_v}-P_G-2u\sum_{w\in K_v}h_wT_{G - \Gamma[w]} T^-_{G-K_v}
		\Bigg)h_v\chi,
		\label{eq:T_Gh_vchiT_G^-}
		\end{equation}
		where we have used a different rearrangement of \cref{eq:PKrecurrence} from that in Eq.~(\ref{eq:chitermline4}) to simplify the expression. 
		We next combine \cref{eq:chitermline4,eq:T_Gh_vchiT_G^-} according to the linear combination on the left-hand side of Eq.~(\ref{eq:fundidentity}) to obtain
		\begin{equation}
		\begin{split}
		T_G\left(1+u\sum_{v\in K_s}h_v\right)\chi T_G^-&=P_G\left(1-u\sum_{v\in K_s}h_v\right)\chi
		\\
		&
		\hspace{5mm}+2u\sum_{v\in K_s}\left(u b_v^2P_{G-\Gamma[v]}-h_v T_{G-\Gamma[v]}T^-_{G-K_s}\right)\chi
		\\
		&
		\hspace{1cm}+2u\sum_{v\in K_s}\left(P_{G-K_v}-u\sum_{w\in K_v}h_w T_{G-\Gamma[w]}T^-_{G-K_v}\right)h_v\chi.
		\end{split}
		\label{eq:lem3proofwDelta}
		\end{equation}
		We next proceed to prove that the last two terms in \cref{eq:lem3proofwDelta} evaluate to zero. 
		Denote these terms by $\mathrm{\Delta}$. 
		What follows is a tedious yet straightforward rearrangement of the expression for $\mathrm{\Delta}$:
		\begin{align}
		\mathrm{\Delta} &= 2u\sum_{v\in K_s}\big(ub_v^2P_{G-\Gamma[v]}-h_v T_{G-\Gamma[v]}T^-_{G-K_s}\big)\chi \nonumber \\ 
		& \mbox{\hspace{5mm}} + 2u\sum_{v\in K_s}\Big(P_{G-K_v}-u\sum_{w\in K_v}h_w T_{G-\Gamma[w]}T^-_{G-K_v}\Big)h_v\chi \label{eq:lem3proofline1}
		\end{align}
		We begin by separating the sum over $w \in K_v$ into the cases where $w = v$ term and $w \neq v$ terms, and make other minor rearrangements for simplification.
		We also commute the operators $\chi$ and $h_v$ to the left of each expression, taking care to keep track of the sign changes and employ $h_v^2 = b_v^2$ in the $w = v$ term. Finally, we have made use of the identification $G - \Gamma[v] \simeq G - K_s - K_v$ for $v \in K_s$ in the subscripts of the first two terms.
		Thus we can write $\mathrm{\Delta}$ as
		\begin{align}
		\mathrm{\Delta}&= 2\chi \Big\{ \Big[ u^2 \sum_{v \in K_s} b_v^2 P_{G - K_s - K_v} + u \Big(\sum_{v \in K_s} h_v T_{G - K_s - K_v}\Big) T_{G - K_s}^-\Big] \nonumber \\
		& \mbox{\hspace{5mm}} - u\sum_{v \in K_s} \Big[P_{G - K_v} h_v + u \Big(\sum_{\substack{w \in K_v \\ w \neq v}} h_v h_w T_{G - \Gamma[w]} + b_v^2 T_{G - K_s - K_v}\Big) T_{G - K_v}^-\Big]\Big\} \label{eq:lem3proofline2} 
		\end{align}
		Next, we expand the factor of $T^-_{G - K_v}$, for the $w = v$ term, by the recursion relation Eq.~(\ref{eq:TKrecurrence}) for the clique $K_s\backslash\{v\}$, so that $\mathrm{\Delta}$ becomes
		\begin{align}
		\mathrm{\Delta} &= 2 \chi \Big\{\Big[u^2 \sum_{v \in K_s} b_v^2 P_{G - K_s - K_v} + u \Big(\sum_{v \in K_s} h_v T_{G - K_s - K_v}\Big)T^-_{G - K_s}\Big] \nonumber \\ 
		&\mbox{\hspace{5mm}} - u \sum_{v \in K_s} \Big[P_{G - K_v}h_v + u \sum_{\substack{w \in K_v \\ w \neq v}} h_v h_w T_{G - \Gamma[w]} T^-_{G - K_v} \label{eq:lem3proofline3} \\
		&\mbox{\hspace{10mm}} + u b_v^2 T_{G - K_s - K_v} \Big(T^-_{G - K_s - K_v} + u \sum_{\substack{w \in K_s \\ w \neq v}} h_w T^-_{G - K_v  - \Gamma[w]}\Big)\Big]\Big\} \nonumber
		\end{align}
		We next simplify the first of the terms in the final parenthesis of \cref{eq:lem3proofline3} expansion using \Cref{lem:transfer_poly} to give
		\begin{align}
		\mathrm{\Delta}&= 2 \chi \Big\{\Big[u^2 \sum_{v \in K_s} b_v^2 P_{G - K_s - K_v} + u \Big(\sum_{v \in K_s} h_v T_{G - K_s - K_v}\Big)T^-_{G - K_s}\Big] \nonumber \\ 
		&\mbox{\hspace{5mm}} - u \sum_{v \in K_s} \Big[P_{G - K_v}h_v + u \sum_{\substack{w \in K_v \\ w \neq v}} h_v h_w T_{G - \Gamma[w]} T^-_{G - K_v} \label{eq:lem3proofline4}\\
		&\mbox{\hspace{10mm}} + u b_v^2 P_{G - K_s - K_v} + u^2 b_v^2 T_{G - K_s - K_v} \sum_{\substack{w \in K_s \\ w \neq v}} h_w T^-_{G - K_v  - \Gamma[w]}\Big]\Big\}. \nonumber
		\end{align}
		Here, we notice that the first term and second-to-last term in \cref{eq:lem3proofline4} cancel to give
		\begin{align} \mathrm{\Delta} &= 2 \chi \Big[u \Big(\sum_{v \in K_s} h_v T_{G - K_s - K_v}\Big)T^-_{G - K_s} - u \sum_{v \in K_s} \Big(P_{G - K_v}h_v + u \sum_{\substack{w \in K_v \\ w \neq v}} h_v h_w T_{G - \Gamma[w]} T^-_{G - K_v}  \nonumber \\
		&\mbox{\hspace{10mm}} + u^2 b_v^2 T_{G - K_s - K_v} \sum_{\substack{w \in K_s \\ w \neq v}} h_w T^-_{G - K_v  - \Gamma[w]}\Big)\Big] \label{eq:lem3proofline5} \end{align}
		We again make minor rearrangements using the factorizations ${P_{G -K_v}=T_{G - K_v} T^-_{G - K_v}}$ and $b_v^2=h_v^2$, together with the fact that the term $\chi h_v$ commutes with all operators outside of the clique $K_v  \subseteq \Gamma[w]$ for $w \in K_v$, so that
		\begin{align}
		\mathrm{\Delta}&= 2 u \chi \sum_{v \in K_s} h_v \Big( T_{G - K_s - K_v} T^-_{G - K_s} - u^2 h_v T_{G - K_s - K_v} \sum_{\substack{w \in K_s \\ w \neq v}} h_w T^-_{G - K_v - \Gamma[w]} \Big) \nonumber \\ 
		&\mbox{\hspace{5mm}} - 2u \sum_{v \in K_s} \Big(T_{G - K_v} - u \sum_{\substack{w \in K_v \\ w \neq v}} h_w T_{G - \Gamma[w]}\Big) T^-_{G - K_v} \chi h_v \label{eq:lem3proofline6} 
		\end{align}
		Making use of \cref{eq:TKrecurrence}, we can rewrite the parentheses in the second term as $T_G$ plus the missing term from the sum over $K_v$, again making use of the identification $G - \Gamma[v] \simeq G - K_s - K_v$ for $v \in K_s$:
		\begin{align}
		\mathrm\Delta&= 2 u \chi \sum_{v \in K_s} h_v \Big(T_{G - K_s - K_v} T^-_{G - K_s} - u^2 h_v T_{G - K_s - K_v} \sum_{\substack{w \in K_s \\ w \neq v}} h_w T^-_{G - K_v - \Gamma[w]} \Big) \nonumber \\ 
		&\mbox{\hspace{5mm}} - 2u \sum_{v \in K_s} \big(T_G + u h_v T_{G - K_s - K_v}\big) T^-_{G - K_v} \chi h_v \label{eq:lem3proofline7}
		\end{align}
		Next, we collect the residual term from the sum over $K_s$ in the second parentheses with the final term in the first parentheses, to obtain
		\begin{align}
		\mathrm{\Delta} &= 2 u \chi \sum_{v \in K_s} \Big[h_v T_{G - K_s - K_v} T^-_{G - K_s} + u b_v^2 T_{G - K_s - K_v} \Big(-u \sum_{\substack{w \in K_s \\ w \neq v}} h_w T^-_{G - K_v - \Gamma[w]} + T^-_{G - K_v} \Big) \Big] \nonumber \\ 
		&\mbox{\hspace{5mm}} - 2u \sum_{v \in K_s} T_G T^-_{G - K_v} \chi h_v \label{eq:lem3proofline8} \end{align}
		Expanding the second term in parentheses using the recurrence relation Eq.~(\ref{eq:TKrecurrence}) for the clique $K_s\backslash\{v\}$, we find
		\begin{align}
		\mathrm{\Delta}&= 2 u \chi \sum_{v \in K_s} \Big[h_v T_{G - K_s - K_v} T^-_{G - K_s} \nonumber \\ 
		& \mbox{\hspace{5mm}} + u b_v^2 T_{G - K_s - K_v} \Big(-u \sum_{\substack{w \in K_s \\ w \neq v}} h_w T^-_{G - K_v - \Gamma[w]} + T^-_{G - K_s - K_v} + u \sum_{\substack{w \in K_s \\ w \neq v}} h_w T^-_{G - K_v - \Gamma[w]} \Big) \Big]  \label{eq:lem3proofline9}
		\\ 
		&\mbox{\hspace{10mm}} - 2u \sum_{v \in K_s} T_G T^-_{G - K_v} \chi h_v \nonumber
		\end{align}
		The first and third terms in parentheses in \cref{eq:lem3proofline9} cancel, and by \cref{lem:operatorcomms}, the remaining term is $P_{G - K_s - K_v}$. Thus
		\begin{align}
		\mathrm{\Delta} &= 2 u \chi \sum_{v \in K_s} \Big(h_v T_{G - K_s - K_v} T^-_{G - K_s} + u b_v^2 P_{G - K_s - K_v} \Big) - 2u \sum_{v \in K_s} T_G T^-_{G - K_v} \chi h_v \label{eq:lem3proofline10}
		\end{align}
		Finally, we employ the recursion relation Eq.~(\ref{eq:TKrecurrence2}) to obtain
		\begin{align}
		\mathrm{\Delta} &= 2 u \chi \sum_{v \in K_s} \left(h_v T_{G - K_s - K_v} T^-_{G - K_s} + u b_v^2 P_{G - K_s - K_v} \right) + 2T_G \left(T^-_G - T^-_{G - K_s} \right) \chi \\
		&= -2 \left(T_G + u \sum_{v \in K_s} h_v T_{G - K_s - K_v}\right) T^-_{G - K_s} \chi + 2u^2 \chi \sum_{v \in K_s} b_v^2 P_{G - K_s - K_v} + 2 T_G T^-_G \chi \\
		&= -2 T_{G - K_s} T^-_{G - K_s} \chi + 2u^2 \chi \sum_{v \in K_s} b_v^2 P_{G - K_s - K_v} + 2 T_G T^-_G \chi \\
		&= -2 \left(P_{G - K_s}  - u^2 \sum_{v \in K_s} b_v^2 P_{G - K_s - K_v}\right) \chi + 2 T_G T^-_G \chi \\
		\mathrm{\Delta} &= -2 P_G \chi + 2 P_G \chi = 0
		\end{align}
		Therefore, Eq.~(\ref{eq:lem3proofwDelta}) evaluates to
		\begin{align}
		T_G\left(1+u\sum_{v\in K_s}h_v\right)\chi T_G^-&=P_G\left(1-u\sum_{v\in K_s}h_v\right)\chi
		\end{align}
		and this proves the lemma. \qed
	\end{proof}
	Since $u_{\ell}$ is a root of $P_G(-u^2)$, the right-hand-side of \Cref{eq:fundidentity} becomes zero for $u=u_\ell$.. The left-hand side can then be rearranged, such that we can rewrite \cref{eq:Hcommchi} as
	\begin{equation}
	[H,T_G(\pm u_{\ell})\chi T_G(\mp u_{\ell})]=\mp\frac{2}{u_{\ell}}T_G(\pm u_{\ell})\chi T_G(\mp u_{\ell}).
	\label{eq:raise_expr}
	\end{equation}
	Thus, \Cref{eq:raise_expr} implies that the incognito modes of the model act as canonical ladder operators and the Hamiltonian of the ECH model can be written as~\cref{eq:ffsolvable} in terms of the incognito modes, $\psi_{\ell}$.
	
	Finally, to show that the canonical modes of the Hamiltonian are fermionic, we must confirm that the incognito modes obey the canonical anticommutation relations.
	It is straightforward to see from the definition of $\psi_\ell$ that ${(\psi_{\ell})^2 \propto P(-u_\ell^2)^2=0}$ (remember the simplicial mode $\chi$ is defined as a Pauli operator, so $\chi^2 = I$).
	Further, since the transfer matrix and the simplicial mode are both Hermitian, we have
	\begin{equation}
	\psi_{\ell}^\dagger=\frac{1}{N_\ell}\left(T(u_\ell)\chi T(-u_\ell)\right)\coloneqq\psi_{-\ell}.
	\end{equation}
	\begin{lemma}
		\label{lem:operatorcomms}
		The incognito modes, $\{\psi_{\ell}\}$ (\cref{def:incognito}),
		satisfy the canonical anticommutation relations,
		\begin{equation}
		\{\psi_\ell,\psi_{-m}\}=\delta_{m,\ell}
		\label{eq:carrestate}
		\end{equation}
		with normalization
		\begin{equation}
		(N_\ell)^2={16u_\ell^2}P_{G-K_s}(-u^2_\ell)\partial_x(P_G(x))_{x = -u_\ell^2},
		\label{eq:normalization}
		\end{equation}
		where $\partial_x(P_G(x))_{x = -u_\ell^2}$ denotes the derivative of $P_G(x)$ with respect to $x$, evaluated at $-u_\ell^2$.
	\end{lemma}
	\begin{proof}
		This proof closely follows the derivation by Fendley~\cite{fendley2019free}, again with generalizations to the graph recursion relations. 
		To find the anticommutator between any two fermionic operators, we take the limit of
		\begin{equation}
		\{\psi_\ell,\psi_{-m}\}=\frac{1}{N_m}\lim_{u\rightarrow u_m}\{\psi_\ell,T_G(u)\chi T_G(-u)\}
		\label{eq:anti_comm_lim}
		\end{equation}
		We start by finding the explicit relationship acquired by commuting $\psi_\ell$ and $T_G(u)$.
		To do this we expand $T_G(u)\chi T_G(-u)$, using \cref{eq:fundidentity,lem:simplicialTcomm}, noting that the transfer matrices commute, even at different $u$, due to \Cref{lem:conserved_charges}, so that
		{\begin{align}
			T_G(u)\psi_\ell T_G(-u)=&\frac{1}{N_\ell}T_G(-u_\ell)\left[-u T_G(u)\left(\sum_{v \in K_s} h_v \chi \right) T_G(-u) + P_G(-u^2) \left(1 - u \sum_{v \in K_s} h_v \right) \chi \right] T_G(u_{\ell}) \\
			=&-\frac{u}{2}T_G(u)[H,\psi_\ell]T_G(-u)+P_G(-u^2)\bigg(\psi_\ell-\frac{u}{2}[H,\psi_\ell]\bigg)\\
			T_G(u)\psi_\ell T_G(-u)=&\frac{1}{u_\ell}\bigg(-uT_G(u)\psi_\ell T_G(-u)+P_G(-u^2)(u_\ell-u)\psi_\ell
			\bigg)
			\label{eq:TpsiT}
			\end{align}}
		By rearranging the expression in \cref{eq:TpsiT}, we find algebra obeyed by the transfer matrices and $\psi_\ell$ as
		\begin{equation}
		(u_\ell+u)T_G(u)\psi_\ell=(u_\ell-u)\psi_\ell T_G(u).
		\label{eq:(ul_u)psi}
		\end{equation}
		Thus, the argument of \Cref{eq:anti_comm_lim} becomes 
		\begin{equation}
		\{\psi_\ell,T_G(u)\chi T_G(-u)\}=\frac{u_\ell+u}{u_\ell-u}T_G(u)\{\psi_\ell,\chi\}T_G(-u).
		\label{eq:anticommratio}
		\end{equation}
		The anticommutator between $\psi_\ell$ and $\chi$ can be calculated explicitly using the recursion relation \Cref{eq:TKrecurrence},
		\begin{equation}
		\{\psi_\ell,\chi\}=\frac{4}{N_\ell}P_{G-K_s}(-u_\ell^2).
		\label{eq:anti_com_psi_chi}
		\end{equation}
		Since the right hand side of \cref{eq:anti_com_psi_chi} is scalar, we can commute the transfer matrix, $T(u)$, in \cref{eq:anti_comm_lim} through the anticommutator and use \Cref{eq:vertex_weighted_poly} to write explicitly
		\begin{equation}
		\{\psi_\ell,T_G(u_m)\chi T_G(-u_m)\}=\lim_{u\rightarrow u_m}\frac{4}{N_\ell}P_{G-K_s}(-u_\ell^2)P_G(-u^2)\frac{u_\ell+u}{u_\ell-u}.
		\end{equation}
		In this limit we find that the polynomial $P_G(-u^2)\rightarrow0$, except in the case when $\ell=m$.
		Here, this limit requires the use of L'H\^{o}pital's rule, since both the numerator and denominator of the expression go to zero.
		Doing so gives
		\begin{eqnarray}
		\{\psi_\ell,\psi_{-m}\}=\delta_{\ell,m}\frac{16u_\ell^2}{N_\ell^2}P_{G-K_s}(-u^2_\ell)\partial_x(P_G(x))_{x=-u_\ell^2}.
		\end{eqnarray}
		Thus, we define the normalization factor of the incognito modes to be
		\begin{equation}
		(N_\ell)^2={16u_\ell^2}P_{G-K_s}(-u^2_\ell)\partial_x(P_G(x))_{x=-u_\ell^2},
		\label{eq:normalization2}
		\end{equation}
		revealing that the $\{\psi_\ell\}$ do indeed satisfy the algebra of fermions.
		\qed
	\end{proof}
	Finally, we prove Theorems \ref{thm:theorem1} and \ref{thm:theorem2}. In general, the existence of fermionic ladder operators satisfying \cref{eq:raise_expr,eq:carrestate} is only enough to show that the Hamiltonian block-diagonalizes into sectors (i.e. multiplets). In each sector, the Hamiltonian has the same free spectrum up to a sector-dependent constant shift. In our case, however, the transfer matrix formalism allows us to prove the stronger statements of  Theorems \ref{thm:theorem1} and \ref{thm:theorem2}. Having proven the necessary lemmas, this proof is straightforward, as it matches exactly to the proof given by Fendley in Ref.~\cite{fendley2019free}. We restate the essential steps of this proof here for completeness.
	
	\noindent \textit{Proof of \cref{thm:theorem1,thm:theorem2}.}
	In Ref.~\cite{fendley2019free}, the \emph{higher Hamiltonians} $\{H^{(k)}\}_{k = 1}^{\infty}$  are defined as operators generated by the logarithmic derivative of $T_G$
	\begin{align}
	\mathcal{H}(u) \coloneqq \sum_{k = 1}^{\infty} H^{(k)} u^{k - 1} \coloneqq -\partial_u \ln \left[T_{G}(u)\right] = -\frac{1}{P_G(-u^2)} T_{G}(-u) T'_{G}(u) \,.
	\end{align}
	The last equality follows from Lemma \ref{lem:transfer_poly}, where $T'_G$ is the derivative of $T_G$ with respect to $u$. Substituting $u = 0$ into the second and fourth expressions in this definition demonstrates that $H^{(1)} \coloneqq H$. The operator $\mathcal{H}(u)$ is a meromorphic function of $u$ whose only singularities are at the roots of $P_G(-u^2)$, since $T_G(-u)T'_G(u)$ is a finite series in $u$ with bounded-operator coefficients. Since $P_G(x)$ has a constant term, none of the roots of $P_G(-u^2)$ are at $u = 0$, and so $\mathcal{H}(u)$ is analytic on a small disk centered at this point. Therefore, we can write each of the higher Hamiltonians as an integral over a small oriented contour $C$ around $u = 0$
	\begin{align}
	H^{(k)} = \frac{1}{2 \pi i} \oint_C du \ u^{-k} \mathcal{H}(u)  \,.
	\end{align}
	Changing variables to $u = 1/\epsilon$ gives
	\begin{align}
	H^{(k)} = \frac{1}{2 \pi i} \oint_{\widetilde{C}} d\epsilon \ \epsilon^{k - 2} \mathcal{H}(\epsilon)
	\end{align}
	where the new contour $\widetilde{C}$ encircles all of the poles at the zeros of $P_G(-u^2)$ with the same orientation as $C$ (reversing this orientation incurs a sign change). This gives
	\begin{align}
	H^{(k)} = -\frac{1}{2\pi i} \oint_{\widetilde{C}} d \epsilon \frac{\epsilon^{2\alpha + k - 2}}{\prod_{j = 1}^{\alpha} (\epsilon^2 - \epsilon^2_j)} T_G(-1/\epsilon) T'_G(1/\epsilon)
	\end{align}
	where we have utilized the factorization $P_G(-u^2) = \prod_{j = 1}^{\alpha} \left(1 - \epsilon_j^2 u^2\right)$ and multiplied numerator and denominator by $u^{-2\alpha} \coloneqq \epsilon^{2\alpha}$ in the integral ($u\neq0$ over $C$). Note here that $T_G'$ is still the derivative of $T_G$ with respect to its argument and not the derivative with respect to $\epsilon$ in the equation above. Since the maximum power of $u$ in $T_G(u)$ is $\alpha$, the minimum power of $\epsilon$ in $T_G(-1/\epsilon) T_G'(1/\epsilon)$ is $-2\alpha + 1$, and so the integrand has no poles at $\epsilon = 0$ for $k \geq 1$. The only poles of the integrand are therefore at $\pm \epsilon_j$, and so the Cauchy residue theorem gives
	\begin{align}
	H^{(k)} &= -\sum_{j = 1}^{\alpha} \frac{\epsilon_j^{2\alpha + k - 2}}{\prod_{\ell = 1, \ell \neq j}^{\alpha} (\epsilon_j^2 - \epsilon^2_{\ell})} \left[\frac{1}{2\epsilon_j} T_G(-1/\epsilon_j) T'_G(1/\epsilon_j) - \frac{(-1)^k}{2\epsilon_j} T_G(1/\epsilon_j) T'_G(-1/\epsilon_j) \right] 
	\end{align}
	Using
	\begin{align}
	\partial_u \left[P_G(-u^2)\right]_{u = u_j} = -2 \epsilon_j \prod_{\ell = 1, \ell \neq j}^{\alpha} (1 - \epsilon^2_{\ell} u^2_j)
	\end{align}
	gives
	\begin{align}
	H^{(k)} &= \sum_{j = 1}^{\alpha} \frac{u_j^{-k}}{\partial_u \left[P_G(-u^2)\right]_{u = u_j}} \left[T_G(-u_j) T'_G(u_j) - (-1)^k T_G(u_j) T'_G(-u_j) \right]
	\label{eq:higherhamiltonianfinal}
	\end{align}
	Next we evaluate the commutator
	\begin{align}
	[\psi_j, \psi^{\dagger}_j] &= [\psi_j, \psi_{-j}] \label{eq:commlimitline1}\\
	&= \frac{1}{N_j} \lim_{u \rightarrow u_{j}} [\psi_j, T_{G}(u) \chi T_{G}(-u)] \label{eq:commlimitline2} \\
	[\psi_j, \psi^{\dagger}_j] &= \frac{1}{N_j} \lim_{u \rightarrow u_{j}} \left\{\left(\frac{u_j + u}{u_j - u}\right) T_{G}(u) [\psi_j, \chi] T_{G}(-u) \right\} \label{eq:commlimitline3}
	\end{align}
	This follows by similar steps to \cref{eq:anti_comm_lim,eq:anticommratio}. Our definition of the incognito modes, together with \Cref{lem:transfer_poly}, implies
	\begin{align}
	T_G(u_j) \psi_j = \psi_{-j} T_{G}(u_j) = 0
	\label{eq:zeroeigenvector}
	\end{align}
	both numerator and denominator of Eq.~(\ref{eq:commlimitline3}) vanish in the limit, so we have
	\begin{align}
	[\psi_j, \psi^{\dagger}_j] &= -\frac{2 u_j}{N_j} \left( T'_G(u_j) \psi_j \chi T_G(-u_j) + T_G(u_j) \chi \psi_j T'_G(-u_j)\right) \,.
	\end{align}
	Exchanging $\chi$ and $\psi_j$ using the anticommutator, Eq.~(\ref{eq:anti_com_psi_chi}), gives
	\begin{align}
	[\psi_j, \psi^{\dagger}_j] &= -\frac{8 u_j}{N^2_j} P_{G - K_s}(-u_j^2) \left( T'_G(u_j) T_G(-u_j) + T_G(u_j) T'_G(-u_j)\right)
	\end{align}
	as the additional terms vanish by Eq.~(\ref{eq:zeroeigenvector}). Rewriting the normalization condition of Eq.~(\ref{eq:normalization}) as,
	\begin{equation}
	(N_j)^2=-8u_jP_{G-K_s}(-u^2_j)\partial_u[P_G(-u^2)]_{u=u_j},
	\end{equation}
	and then substituting into the above expression gives
	\begin{align}
	[\psi_j, \psi^{\dagger}_j] &= \frac{1}{\partial_u[P_G(-u^2)]_{u=u_j}} \left( T'_G(u_j) T_G(-u_j) + T_G(u_j) T'_G(-u_j)\right)
	\label{eq:commutatorref}
	\end{align}
	Comparing Eq.~(\ref{eq:higherhamiltonianfinal}) for $k = 1$ to Eq.~(\ref{eq:commutatorref}) proves both Theorems \ref{thm:maintheorem1} and \ref{thm:maintheorem2}. \qed
	
	In this way we can see that a ECH model is described by noninteracting fermions, with single particle energies given by the reciprocals of the roots of the vertex-weighted independence polynomial, \cref{eq:vertex_weighted_poly}, and canonical modes given by the incognito modes (\cref{def:incognito}).
	This constitutes a complete solution to any model of this kind.
	
	\section{Examples}
	\label{sec:examples}
	In this section we analyze explicitly three sets of models whose Hamiltonians are (even-hole, claw)-free, thus admitting a free-fermion solution via \Cref{thm:maintheorem2,thm:maintheorem1}.
	The first set includes a pair of examples of models realized on small systems, including a model whose frustration graph is a line graph and a simple extension whose frustration graph is one of the nine forbidden subgraphs of a line graph (see~\Cref{tab:forbidden_lines}).
	The second set of examples is a class of graphs we call \textit{equipartition indifference} graphs.
	This family of models generalizes the well-known XY-model and was exactly solved at their critical points in Refs.~\cite{Alcaraz2020Free, Alcaraz2020Integrable}.
	In the third set of examples, we define a new family of integrable models constructed by combining the equipartition indifference graphs into more complex structures; these may also be solvable if they avoid even holes.
	
	\subsection{Small systems}
	\label{sec:smallsystems}
	Here we look at two related models on three qubits.
	The Hamiltonians of the models, denoted $H_5$ and $H_6$, and are related by the addition of a single term 
	\begin{align}
	H_5&=aX_1X_2+bZ_2+cY_1Y_2X_3+dY_1Z_2+eX_1Z_2\\
	H_6&=aX_1X_2+bZ_2+cY_1Y_2X_3+dY_1Z_2+eX_1Z_2+fY_1Y_2Z_3,
	\end{align}
	where the coupling strengths $\{a, b, c, d, e, f\}$ are arbitrary real numbers. 
	The frustration graphs for $H_5$ and $H_6$ are depicted in \Cref{fig:circle_graph}, with vertices labeled by their corresponding field strengths. 
	
	The frustration graph of the first model, $G(H_5)$, is a five-cycle, as depicted in~\cref{fig:circle_graph}~(a): the five Hamiltonian terms anticommute only with those directly before and after it in a closed chain. 
	$G(H_5)$ is a line graph, as such this model admits a Jordan-Wigner solution~\cite{Chapman2020Characterization}.
	Let $R(H_5)$ be the root graph of $G(H_5)$; this graph is also a five-cycle. Each Hamiltonian term is mapped to a Majorana bilinear with Majorana modes $\{\gamma_j\}_{j = 1}^5$ assigned to each vertex of the root graph. 
	Using this method, the Hamiltonian is mapped to
	\begin{equation}
	H= \frac{i}{2} \sum_{j,k=0}^4 \gamma_j h_{j,k} \gamma_k
	\end{equation}
	where the single particle Hamiltonian is given by
	\begin{equation}
	\frac{i}{2} \mathbf{h}= \frac{i}{2}\begin{pmatrix}
	0 & a& 0& 0&-e\\
	-a& 0& b& 0& 0\\
	0 &-b& 0& c& 0\\
	0 & 0&-c& 0& d\\
	e & 0& 0&-d& 0
	\end{pmatrix}.
	\end{equation}
	Note here that the orientation of $R(H_5)$ given by the signs of the elements in $\mathbf{h}$ is arbitrary, i.e. we can change the sign of any coupling coefficient without affecting the spectrum.
	
	\begin{figure}
		\centering
		\includegraphics{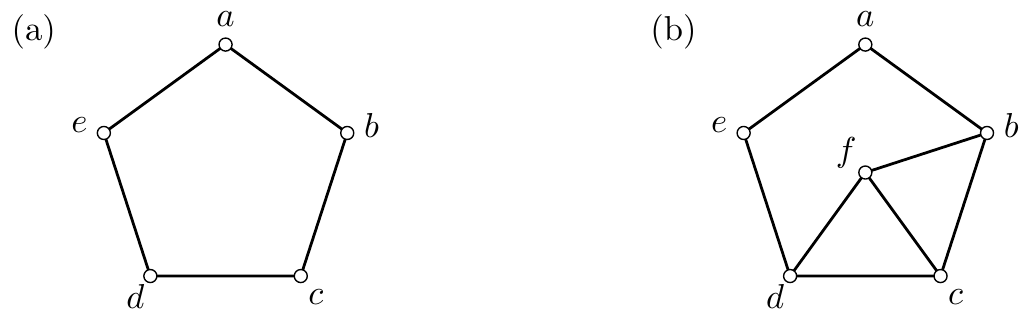}
		\caption{Frustration graphs for small system sizes solved in this section: (a) a five-cycle, which admits a generator-to-generator mapping, but is both even hole and claw free also, thus admitting a solution by the method developed in the present work,
			(b) one of the six forbidden subgraphs of a line graph with no twins, created by adding a single additional Hamiltonian term to the five-cycle, thus when $f\rightarrow0$, this model is identical to (a). This model does not admit a generator-to-generator mapping, but is solvable using the method developed here.}
		\label{fig:circle_graph}
	\end{figure}
	
	The frustration graph of $G(H_5)$ is also an (even hole, claw)-free graph: every edge induces a simplicial clique.
	Thus, the model can be solved using the method developed here.
	The vertex-weighted independence polynomial of $G(H_5)$ is
	\begin{equation}
	\begin{split}
	P_{G(H_5)}(-u^2)&=1-u^2 \left(a^2+b^2+c^2+d^2+e^2\right)\\
	&\hspace{1cm}+u^4 \left[a^2 \left(c^2+d^2\right)+b^2 \left(d^2+e^2\right)+c^2e^2\right].
	\end{split}
	\label{eq:poly_five_vert}
	\end{equation}
	$P_{G(H_5)}$ can be factored simply as a quadratic polynomial in $u^2$. The roots of $P_{G(H_5)}$ provide the the single particle energies, as well as the spectral parameters for the incognito modes, $\{\psi_k\}_k$ in Eq.~(\ref{eq:incognitodef}).
	Note that, as discussed in \cref{sec:symmetries}, $P_{G(H_5)}$ is exactly the characteristic polynomial of the single-particle Hamiltonian $i\mathbf{h}$
	\begin{align}
	P_{G(H_5)}(-u^2) = \det \left(\mathbf{I} - i u \mathbf{h} \right) \,.
	\end{align}
	Thus, we can see the direct link between the two approaches for a solution when the model is an even hole-free line graph.
	Furthermore, the eigenvectors of the single particle Hamiltonian, $\mathbf{h}$, elucidate the nonlocality of the canonical modes, $\{\psi_k\}_k$.
	
	The frustration graph $G(H_6)$ is depicted in \cref{fig:circle_graph}~(b).
	In direct contrast to $H_5$, this graph is one of the six forbidden subgraphs of a line graph that does not contain twins and admits no Jordan-Wigner mapping to noninteracting fermions.
	Nevertheless, \cref{fig:circle_graph}~(b)  contains no even holes or claws, and each maximal clique of the graph is simplicial.
	Thus, by \Cref{thm:theorem1} the model must be free.
	
	The vertex-weighted independence polynomial of $G(H_6)$ is
	\begin{equation}
	\begin{split}
	P_{G(H_6)}(-u^2)=&1-u^2 \left(a^2+b^2+c^2+d^2+e^2+f^2\right)+\\
	&+u^4\left[a^2\left(c^2+d^2+f^2\right)+b^2\left(d^2+e^2\right)+ c^2e^2+e^2f^2\right]
	\end{split}
	\label{eq:poly_six_vert}
	\end{equation}
	and so the single particle energies can be found by again solving a simple quadratic equation.
	
	Despite the similarities between the two models, $H_5$ admits a solution in terms of individual fermions localized to physical modes, while $H_6$ does not. 
	It therefore remains an open question to clarify the intrinsic link (if any) between the graphical and spatial structures for models with ECF frustration graphs, a stark contrast from the line-graph setting.
	
	While the spectrum, and fermionization, of the models is independent of the explicit Pauli realization, the qubitization of the graphs given does elucidate an interesting link between the ECF models and those that have a Jordan-Wigner mapping (line-graph models).
	Here $H_5$ and $H_6$ are related via a single-qubit rotation on the third qubit.
	To see this, rewrite the Hamiltonian as
	\begin{equation}
	H_6=aX_1X_2+bZ_2+Y_1Y_2(cX_3+fZ_3)+dY_1Z_2+eX_1Z_2.
	\end{equation}
	By applying the coupling strength dependent rotation $(cX_3+fZ_3)\rightarrow \pm \sqrt{c^2 + f^2} X_3$, we can see the direct relation between models $H_5$ and $H_6$. 
	In general, the transformation from an arbitrary (even-hole, claw)-free model to a similar line-graph model is nontrivial, requiring complicated, multi-qubit rotations.
	However, this particular example shows when two vertices share the same \emph{closed neighborhood} we can always perform a rotation to remove one of them, without altering the spectrum. This is analogous to the situation involving twin vertices -- i.e. vertices sharing an open neighborhood -- discussed in \cref{sec:symmetries}. The difference is that vertices sharing a closed neighborhood are themselves neighboring. Similarly to removing twin vertices by projecting onto a subspace, we remove these pairs by performing a rotation.

	\begin{figure}[t]
		\centering
		\begin{tabular}{cc}
			(a)\ \ \ \begin{minipage}{.45\textwidth}
				$k=2$ \includegraphics[width=.9\textwidth]{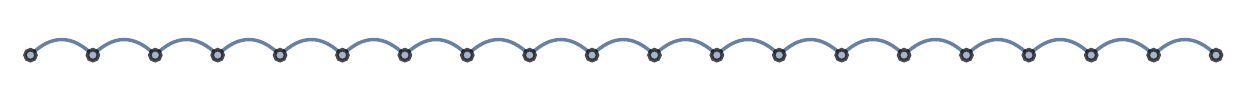}
				$k=3$ \includegraphics[width=.9\textwidth]{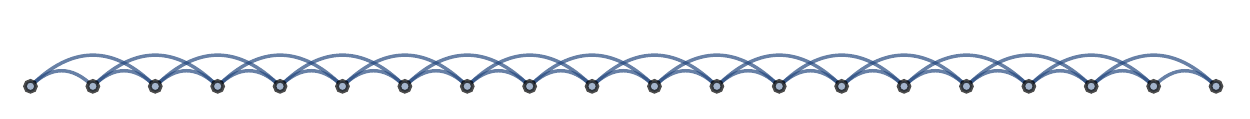}
				$k=4$ \includegraphics[width=.9\textwidth]{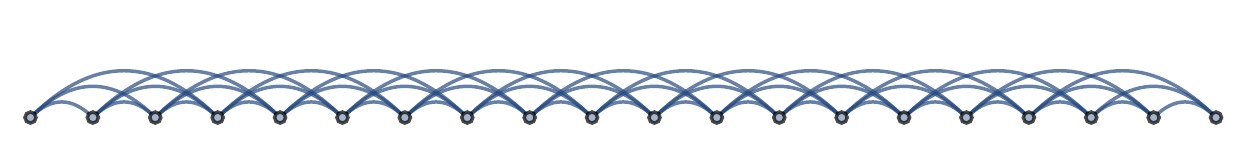}
				$k=5$ \includegraphics[width=.9\textwidth]{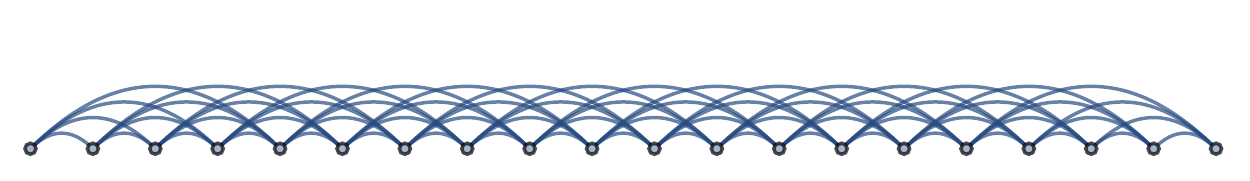}
			\end{minipage}
			& 
			\begin{minipage}{.45\textwidth}
				(b) \includegraphics[width=.9\textwidth]{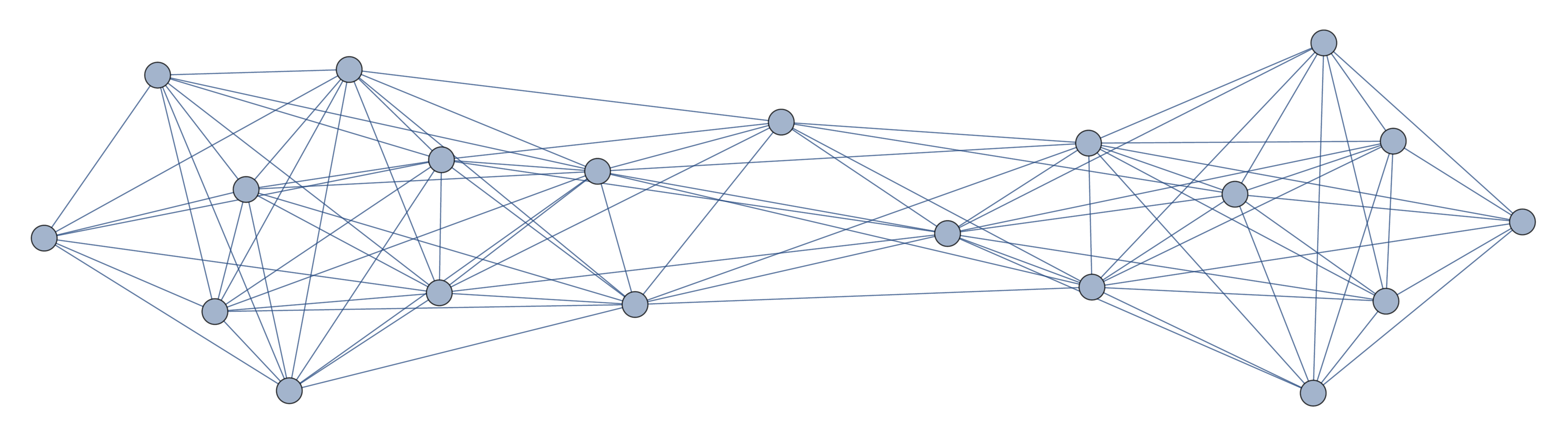}
				(c) \includegraphics[width=.9\textwidth]{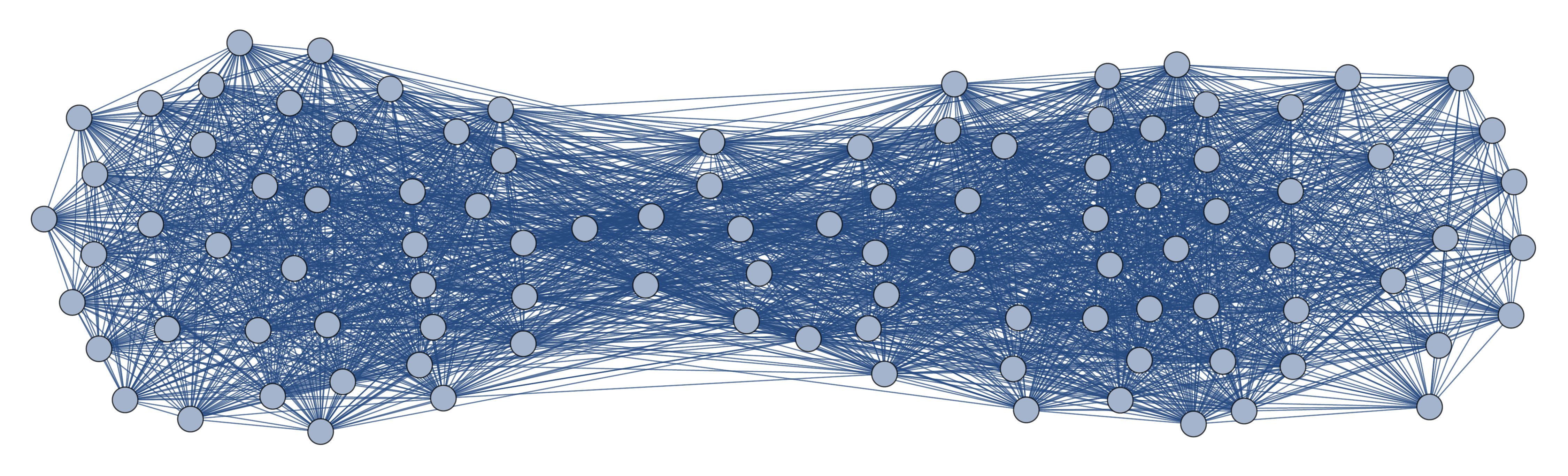}
			\end{minipage}
		\end{tabular}
		\caption{Indifference frustration graphs for some ECF models (see text). 
			\textbf{(a)}~Equipartition indifference graphs arise, for example, as frustration graphs for the translation-invariant spin chain (having open boundary conditions) with Hamiltonian $H_k = \sum_i X_i X_{i+1} \cdots X_{i+k-2} Z_{i+k-1}$. 
			The frustration graph is shown for $k=2,\ldots,5$. 
			When $k=2$, a Hadamard rotation on every second spin shows that this model is equivalent to the XY model with half of its terms removed, and it can be solved by a Jordan-Wigner transformation. 
			The $k=3$ model is the model studied by Fendley~\cite{fendley2019free}. 
			When $k> 2$, the graph is not a line graph, so it cannot be solved by any Jordan-Wigner transformation~\cite{Chapman2020Characterization}, but since it is ECF by construction, it can be solved by the methods introduced in this paper. 
			\textbf{(b,c)}~Two examples of ECF graphs that arise as indifference graphs with randomly chosen points at different densities. 
			It is evident that, depending on the point density, the connectivity of interval graphs can look rather complex.}
		\label{fig:XXXXZ}
	\end{figure}

	\subsection{Indifference Graphs}
	\label{sec:indifference}
	An infinite family of ECF graphs is given by the set of \textit{indifference graphs}. 
	Indifference graphs are defined by placing vertices on the real line and connecting two vertices if and only if they are separated by a distance $\le k$, for some fixed and finite $k$. 
	Such graphs are ECF since they have a known forbidden induced subgraph characterization that forbids (among others) the claw and even holes~\cite{wegner1967eigenschaften}. 
	They are therefore also simplicial. 
	In fact, the closed neighborhood of the vertex corresponding to the least real number is always a simplicial clique (its neighbors are all within distance 1 of each other, hence induce a clique). 
	Finally, it is simple to identify the independent sets for these graphs: they are the subsets of vertices whose pairwise separation on the real line is greater than 1.
	Some examples of indifference graphs are shown in \cref{fig:XXXXZ}. 
	Given an indifference graph $G$, there is no unique way to find a spin model Hamiltonian having $G$ as its frustration graph, though such models will always exist. 
	To get a natural mapping to spin models, we will specialize to the set of graphs (shown in \cref{fig:XXXXZ}(a)) where the vertices are equally spaced on the real line. 
	
	We therefore consider a particularly nice family of spin models, which generalizes the XY-model and the four-fermion model in Ref.~\cite{fendley2019free}.
	This family was originally introduced in Refs.~\cite{Alcaraz2020Free, Alcaraz2020Integrable} and the critical behavior analyzed there as well. 
	Here we demonstrate how this family fits into our formalism.
	Each model in the family is indexed by an integer $k$.
	When $k=2$, we get an XY-chain, albeit with half the terms removed. 
	This model is still solvable by a Jordan-Wigner transformation. 
	When $k=3$, we get the four-fermion model solved in Ref.~\cite{fendley2019free}.
	When $k \geq 4$, we get an infinite family of free-fermion-solvable models with translation-invariant frustration graphs. 
	The construction of the associated frustration graph with $N$ unit cells, $G(N, k)$, is simple: fix $k$, and consider the set of integers $M(N, k) \subset \mathds{Z}$
	\begin{align}
	M(N, k)  \coloneqq\bigcup_{n = 0}^{N - 1} \bigcup_{j = 0}^{k - 1} \left(n k + j \right) \,.
	\end{align}
	Let $m(n, j) \coloneqq nk + j$. 
	Associate a vertex of $G(N, k)$ to each point in $M(N, k)$, and join vertices corresponding to $m(n, j)$ and $m(n', j')$ by an edge if $|m(n, j) - m(n', j')| < k$. Then $G(N, k)$ is equivalent to the indifference graph of $M(N, k)$ after rescaling our distance function appropriately. 
	We will often refer to the vertices of $G(N, k)$ by their corresponding points in $M(N, k)$ directly. 
	We will shortly see that $N = \alpha$, the independence number of $G(N, k)$. 
	See \cref{fig:astral_triple_free4} for an example when $k = 4$. 
	
	\begin{figure}
		\centering
		\includegraphics[width=0.65\textwidth]{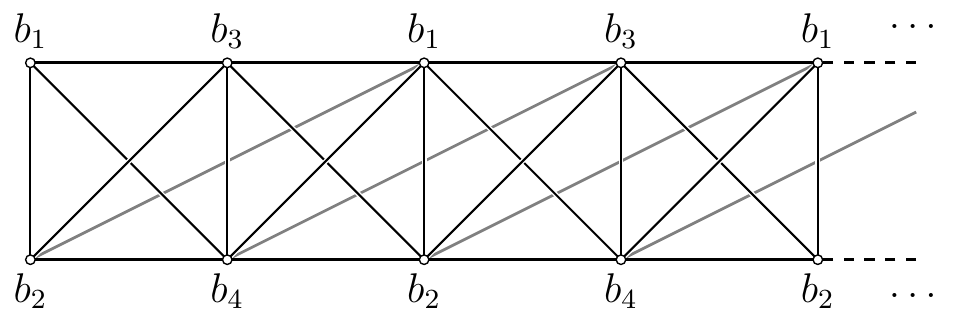}
		\caption{Equipartition indifference graph for $k=4$, formed from the frustration graph of the Hamiltonian in \cref{eq:pauli_equi_ham}.  
			The Hamiltonian couplings are 4-periodic and are labeled $b_1, b_2, b_3$, and $b_4$.}
		\label{fig:astral_triple_free4}
	\end{figure}
	
	An explicit qubit Pauli Hamiltonian realizing $G(N, k)$ is given in Refs.~\cite{Alcaraz2020Free, Alcaraz2020Integrable} 
	\begin{align}
	H =\sum_{n = 0}^{N - 1} \sum_{j = 0}^{k - 1} b_{nk + j} X_{n k + j} \prod_{\ell=1}^{k - 1}Y_{n k + j + \ell}\,.
	\label{eq:pauli_equi_ham}
	\end{align}
	Similarly to these references, we consider staggered, uniform couplings: $k$ different couplings which are repeated periodically as
	\begin{equation}
	b_{nk + j} \coloneqq b_j,
	\end{equation}
	For simplicity of expression we collect the squares of the coupling strengths in a vector
	\begin{align}
	\boldsymbol{b}=(b_0^2,b_1^2,b_2^2,...,b_{k - 1}^2).
	\end{align}
	Define the \emph{elementary symmetric polynomials} in $\boldsymbol{b}$ as
	\begin{align}
	e_j(\boldsymbol{b}) \coloneqq \sum_{0 \leq i_1 < i_2 < \dots < i_j \leq k - 1} \prod_{\ell = 1}^{j} b^2_{i_\ell}
	\end{align}
	for $j \in \{0, \dots, k\}$, with $e_0 \coloneqq 1$. Finally, denote the clique induced by the vertices corresponding to the points $\{m(n, 0), m(n, 1), \dots, m(n, k - 1)\} \subset M(N, k)$ in $G(N, k)$ by $K_n$. 
	Notice that 
	\begin{align}
	G(N, k) - \left(\sum_{p = 0}^{\ell} K_p\right) = G(N - \ell - 1, k)\,.
	\label{eq:graphconvention}
	\end{align}
	Since any independent set can contain at most one vertex from each clique $K_n$, we have that $\alpha \leq N$. An explicit independent set with $N$ vertices is given by $\cup_{n = 0}^{N - 1} m(n ,0)$. Therefore, $\alpha = N$. 
	
	Let us first show that $P_{G(N, k)}$ satisfies a recursion relation which is symmetric in the entries of $\boldsymbol{b}$, which follows from the graph-theoretic recurrence relations.
	\begin{align}
	P_{G(N, k)}  = P_{G(N - 1, k)} - \sum_{\ell = 1}^{k} u^{2\ell} e_{\ell} (\boldsymbol{b}) P_{G(N - \ell, k)} 
	\label{eq:symrecursion}
	\end{align}
	For $k = 4$, for example
	\begin{align}
	P_{G(N, 4)}  &= \left[1 - u^2 e_{1} (\boldsymbol{b})\right] P_{G(N - 1, 4)} - u^4 e_2 (\boldsymbol{b}) P_{G(N - 2, 4)} - u^6 e_3 (\boldsymbol{b}) P_{G(N - 3, 4)} - u^8 e_4 (\boldsymbol{b}) P_{G(N - 4, 4)} \,.
	\end{align}
	We show Eq.~(\ref{eq:symrecursion}) by first expanding $P_{G(N, k)}$ via the recursion relation Eq.~(\ref{eq:PKrecurrence}) in the clique $K_0$, with the convention in Eq.~(\ref{eq:graphconvention}). 
	Note that the neighbors to each vertex $m(0, j) \in K_0$, besides $K_0$ itself, are given by translations $\{m(1, \ell)\}_{\ell = 1}^{j - 1}$. 
	This gives,
	\begin{align}
	P_{G(N, k)} = P_{G(N - 1, k)} - u^2 \sum_{j = 0}^{k - 1} b_j^2 P_{G(N - 1,k) - \sum_{\ell = 0}^{j - 1} m(1, \ell)} \,.
	\label{eq:k0recursion}
	\end{align}
	We can rearrange similar expansions in the induced subgraphs of $K_1$, which are also cliques, to obtain
	\begin{align}
	P_{G(N - 1, k) - \sum_{\ell = 0}^{j - 1} m(1, \ell)} = P_{G(N - 1, k)} + u^2 \sum_{\ell = 0}^{j - 1} b_{\ell}^2 P_{G(N - 2, k) - \sum_{p = 0}^{\ell - 1} m(1, p)}
	\label{eq:k0indrecursion}
	\end{align}
	for $j \in \{1, \dots, k - 1\}$. Substituting Eq.~(\ref{eq:k0indrecursion}) into Eq.~(\ref{eq:k0recursion}) gives
	\begin{align}
	P_{G(N, k)} = \left[1 - u^2 e_1 (\boldsymbol{b})\right] P_{G(N - 1, k)} - u^4 \sum_{j = 0}^{k - 1} \sum_{\ell = 0}^{j - 1}  b_j^2 b_{\ell}^2 P_{G(N - 2, k) - \sum_{p = 0}^{\ell - 1} m(1, p)}
	\label{eq:k0recursionline2}
	\end{align}
	We can iterate this procedure by substituting \cref{eq:k0indrecursion}, with $G(N - 1, k)$ replaced by $G(N - 2, k)$, back into the sum over $j$ and $\ell$ in \cref{eq:k0recursionline2}. 
	Notice that each time we do this, the sum over single vertices $p$ in the subscript of the summand contains one fewer term, the equipartition indifference graph in this subscript contains one fewer clique $K_n$, and the coefficient in the summand acquires another factor from $\boldsymbol{b}$. 
	Iterating $k - 1$ times gives the desired recurrence relation, \cref{eq:symrecursion}.
	
	\begin{figure}
		\centering
		\includegraphics[scale=0.7]{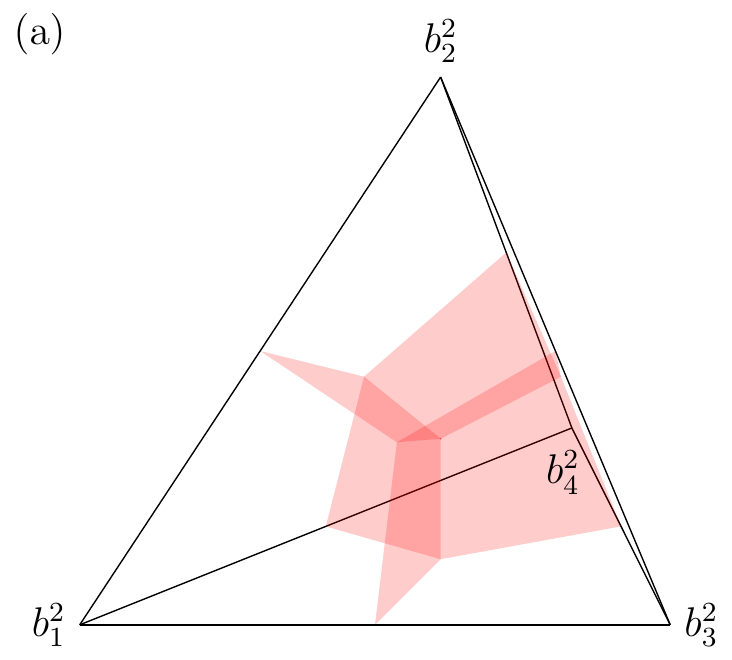}
		\\
		\includegraphics[scale=0.7]{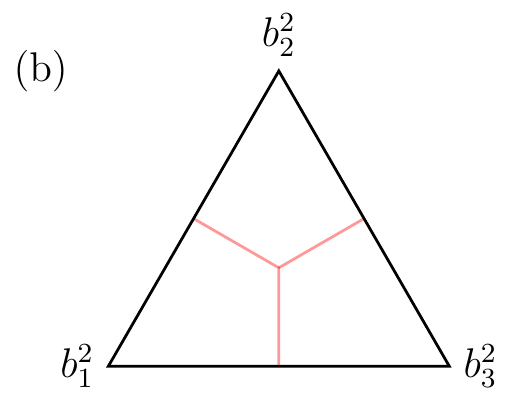}
		\includegraphics[scale=0.7]{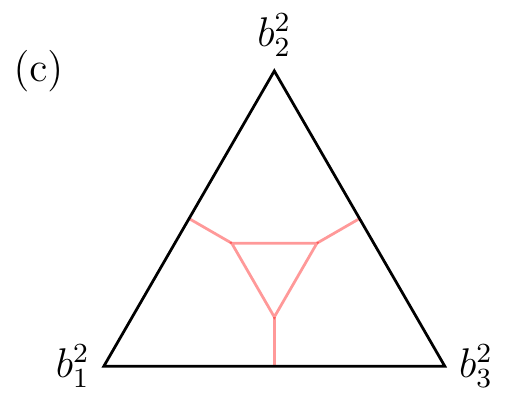}
		\includegraphics[scale=0.7]{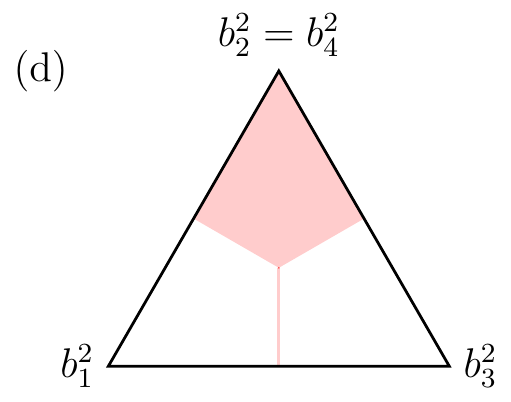}
		\caption{Phase diagram for the \textit{equipartition indifference graph} Hamiltonian for $k=4$, with staggered uniform couplings, $b_1,b_2,b_3,b_4$. 
			(a) shows a three dimensional simplex, with parameters. 
			(b) shows a cross section of the plane at $b_4=0$. 
			(c) shows a cross section parallel to that in (a), but with $b_4>b_1,b_2,b_3$, (d) shows a cross section at a diagonal through the tetrahedron where $b_2^2=b_4^2$ at all times}
		\label{fig:phase_diag}
	\end{figure}
	
	We next assemble a vector with entries
	\begin{equation}
	v_s(\epsilon^2)=\epsilon^{2s}P_{G(s - 1, k)}(\epsilon^{-2}),
	\label{eq:vector}
	\end{equation}
	such that the recursion relation, \cref{eq:symrecursion}, can be rewritten as
	\begin{align}
	v_{N + 1}  &= \epsilon^2 v_N - \sum_{\ell = 1}^{k} e_{\ell} (\boldsymbol{b}) v_{N - \ell + 1} \,.
	\label{eq:vectorrecursion}
	\end{align}
	As this recursion relation holds for any value of $N$, we can define the matrix $\boldsymbol{\mathcal{R}}$ with elements
	\begin{align}
	\mathcal{R}_{ss'}=\sum_{\ell=0}^k\delta_{s-\ell+1,s'}e_{\ell}(\boldsymbol{b}) \mathrm{,}
	\end{align}
	such that \cref{eq:vectorrecursion} has the form of an eigenvalue equation
	\begin{align}
	\boldsymbol{\mathcal{R}} \cdot \boldsymbol{v} = \epsilon^2 \boldsymbol{v} \,.
	\label{eq:eigen_value}
	\end{align}
	When the eigenvalue corresponds to a root $\varepsilon_j^{-2}$ of $P_{G(N, k)}$, the corresponding eigenvector $\boldsymbol{v}$ satisfies the boundary condition $v_{N + 1}(\epsilon_j^2) = 0$. 
	We further require $\boldsymbol{v}$ satisfy the boundary conditions $v_{0}~=~\dots~=~v_{-k + 2}~=~0$ (by our convention, $v_1 \propto P_{G(0, k)} = 1$).
	
	These models exhibit critical behavior when any subset of the coupling coefficients become equal. 
	For all $k$-sized equipartition indifference models, the phase diagram is a $(k-1)$- dimensional simplex with $k$-critical point at the center. 
	The phase diagram for $k=4$ is shown in \cref{fig:phase_diag} as both a three dimensional tetrahedron, as well as three cross sections of the depicting the gapped regions in white, with gapless regions in red.
	Here we see the two-dimensional semi-hyperplanes meeting at the center point of the tetrahedron, where $b_1^2=b^2_2=b^2_3=b^2_4$, as well as along one dimensional lines. 
	It is clear from the cross section in \cref{fig:phase_diag} (b) that the class of models is hereditary as boundary of the phase diagram corresponds directly to the phase diagram of the $k=3$ model (see Ref.~\cite{fendley2019free}).
	Interestingly, \cref{fig:phase_diag} (c) shows that as we increase the fourth parameter, $b_4^2$, the central, tri-critical point in the model opens and a gapless phase emerges.
	\cref{fig:phase_diag} (d) shows the cross-section through the center of the tetrahedron when $b_2^2=b_4^2$. 
	We see that there is a regime in which there is a large gapped phase, as well as two symmetric gapless phases separated by a gapped region.

	The critical behavior of these models has been exactly analyzed in Refs.~\cite{Alcaraz2020Free, Alcaraz2020Integrable} (and indeed, extended to parafermionic systems as well), and the authors find a dynamical critical exponent of $k/d$ for general qudits of dimension $d$ ($d = 2$ in our setting). 
	We add that the model can be numerically analyzed over the entire phase diagram using the fact that $\boldsymbol{\mathcal{R}}$ is a banded Toeplitz matrix with bandwidth $k+1$ and applying the algorithm in Ref.~\cite{Beam1993Asymptotic} to find the dispersion relation in the asymptotic limit.
	\Cref{fig:dispersion_relation} shows the dispersion relations along the central axis from one vertex  ($b_4^2 = 1$) to the center of the opposite face $(b_4^2 = 0)$, with other coefficients equal and normalized. 
	We find that the model is indeed critical when $b_4^2 \leq 0.25$, and the dispersion relation is nonlinear about this point.

	\begin{figure}
		\centering
		\includegraphics[width=0.7\textwidth]{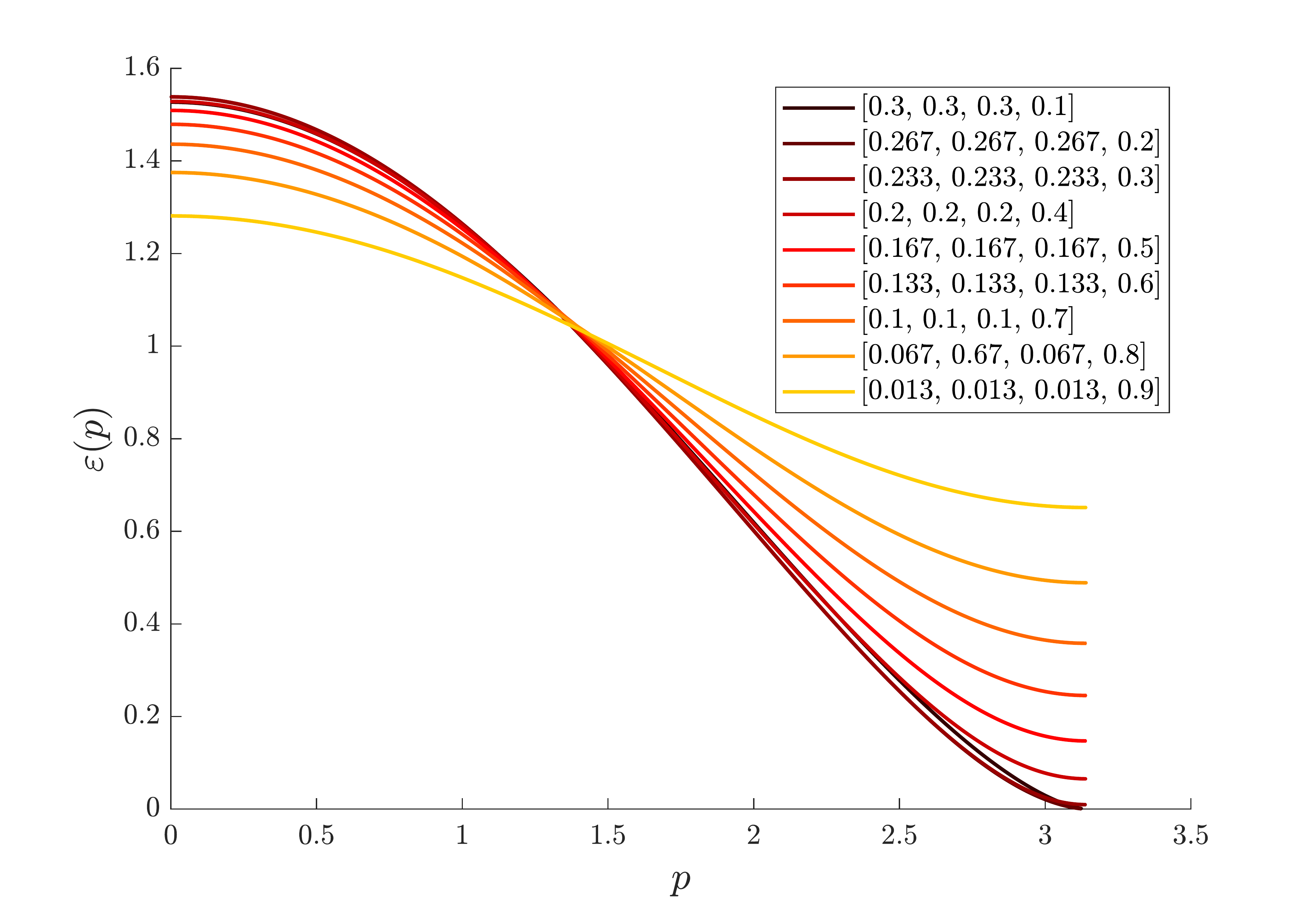}
		\caption{Dispersion relations for the $k = 4$ instance of equipartition indifference models, with $b_4^2 \in \{0.1, 0.2, \dots, 0.9\}$, and other couplings equal such that the sum of their squares is normalized to 1. 
			For $b_4^2 \leq 0.25$, the model is critical at momentum $p = \pi$, as our numerics indicate.}
		\label{fig:dispersion_relation}
	\end{figure}

	\subsection{Integrable models}
	Here we define a new family of two-dimensional models not previously discussed in the literature.
	The models are formed by attaching the one-dimensional chains from \cref{sec:indifference} to one another, via interaction terms with clique-like frustration graphs.
	It is straightforward to define a spin model satisfying such a frustration graph.
	In order to ensure that the graphs remain claw-free, at each point of attachment, or junction, the clique must contain at least twice as many vertices as attached chains. This ensures that the neighborhood of every vertex within the junction clique induces at most two possibly neighboring cliques (bi-simplicial).
	Consider the example frustration graph depicted in \cref{fig:frust_graph1}.
	The junction is trivalent; however, in order to ensure that the model remains claw-free, the junction consists of a clique of six vertices ($K_6$).
	
	Notice that if the joining of the chains is tree-like in its coarse topology, the model will be even-hole-free and thus free-fermion solvable using the methods developed here.
	We imagine such tree-like structures may be useful for probing quantum scrambling.
	On the contrary, if the new structure is genuinely two-dimensional, then the frustration graph will necessarily contain even-holes. Nevertheless, the model will still be integrable, due to \cref{lem:conserved_charges}.
	
	\begin{figure}
	    \centering
	    \includegraphics[width=0.5\textwidth]{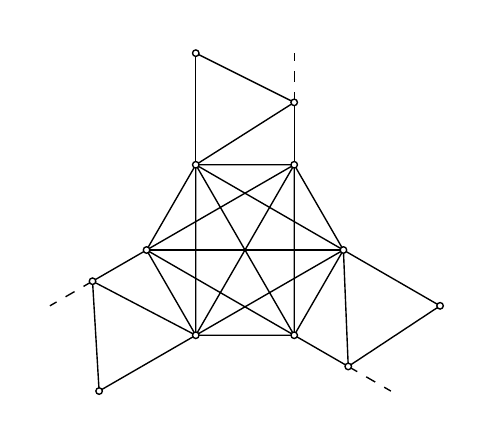}
	    \caption{An example frustration graph of a junction for more complex models which are either integrable or solvable.}
	    \label{fig:frust_graph1}
	\end{figure}
	
	\section{Discussion}
	\label{sec:discussion}
	We have proven that Hamiltonians with (even-hole, claw)-free frustration graphs in a given basis admit a solution by noninteracting fermions, even when such models provably do not admit a Jordan-Wigner solution. 
	Though our result considerably expands the set of known free-fermion solutions, we should note that it clearly does not capture all of them. 
	First, there exist models whose free-fermion solution is \emph{non-generic}, in that they only hold for specific values of the coupling strengths. 
	As an example, we can consider the following model on three qubits: 
	\begin{equation}
	H=aZ_2+bY_1X_2+cX_1Y_2+dZ_1Y_3+eY_1X_3+fZ_3
	\label{eq:back_to_back}
	\end{equation}
	The frustration graph of the model is depicted in~\cref{fig:back_to_back}. 
	The graph  clearly  contains  both  claws  and  even holes, and is thus outside the class of models discussed in this paper.
	In general, the model is not free for arbitrary couplings. 
	Nevertheless, the model does have a free spectrum for all equal couplings ($a=b=c=d=e=f$).
	Further, we can numerically verify that the single particle energies of the model are the reciprocals of the roots of the vertex weighted independence polynomial of the graph.
	
	A perhaps much deeper open question concerns the relationship between the spatial structure of a given model and the associated free-fermion modes which emerge from this solution. 
	We have structured our argument to draw a parallel between the way in which these mappings generalize the Jordan-Wigner transformation---whose spatial structure is evident---with the way claw-free graphs generalize line graphs.
	Carrying this argument through, one might ask whether the even-hole-free assumption can be relaxed, as Hamiltonians whose frustration graphs are arbitrary line graphs still admit a Jordan-Wigner free-fermion solution. 
	From a technical perspective, simplicial claw-free graphs enjoy many of the properties that we relied on to prove our general solution. 
	Models with simplicial claw-free frustration graphs admit an extensive number of commuting conserved charges defined through their independent sets (i.e., they satisfy \Cref{lem:conserved_charges}), and their independence polynomials are also real-rooted \cite{Chudnovsky2007Roots, Engstrom07inequalitieson, Leake2016Generalizations}. 
	One might attempt to incorporate even holes into this formalism by defining spatial hopping terms $\psi_j \psi_k^{\dagger}$. 
	The simplicial mode cancels in the definition of these quadratic operators, leaving them defined only in terms of Hamiltonian terms. 
	However, the resulting expressions are very complicated. 
	Furthermore, allowing for even holes requires us to generalize \Cref{lem:transfer_poly}, which was crucial for the following proof. 
	Though we cannot say anything definitive about the more general class of simplicial claw-free graphs currently (indeed, they may not admit a free-fermion solution at all), we remark that they would be a natural class of models for further study. 
	
	\begin{figure}[t]
		\centering\includegraphics{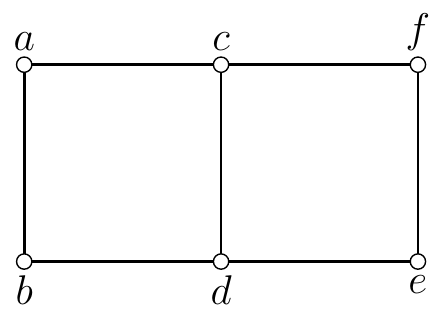}
		\caption{The frustration graph of the model, defined in \Cref{eq:back_to_back}. This model has claws and even holes, nevertheless for all-equal-coefficients the model has a free fermion spectrum.}
		\label{fig:back_to_back}
	\end{figure}
	
	Another clear open question concerns whether this construction could be generalized to solutions of qudit models in terms of \textit{parafermions}~\cite{jaffe}.
	The concept of free parafermions has been developed by Fendley~\cite{Fendley2014Free}, and Refs.~\cite{Alcaraz2020Integrable,Alcaraz2020Free} consider non-Hermitian qudit generalizations of the equipartition indifference graphs which have free parafermionic spectra.
	It is known that, unlike fermions, bilinear parafermions are not always free, yet our formalism may provide a clue to recognizing such systems. 
	Recall that the structure theorem of Ref.~\cite{Cameron2017On} states that (even-hole, pan)-free graphs (which generalize our class of graphs) are essentially unit circular-arc graphs connected by clique cutsets.
	Incidentally, a subset of bilinear parafermion models have frustration graphs given by oriented unit circular-arc graphs, where the orientation captures the fact that the group-commutator between bilinear parafermionic terms is a complex phase, and this orientation is inherited from an underlying orientation on the circular-arc representation. 
	This characterization could therefore clarify the relationship between free-parafermion models, the models considered in Refs.~\cite{Alcaraz2020Integrable,Alcaraz2020Free}, and more general bilinear parafermion models. 
	
	One strength of the solution method developed in the present work is that it could in principle be applied to \emph{interacting} fermion models in addition to qubit models.
	Indeed, Fendley's four-fermion model is an obvious example. 
	This is because the existence of the solution is independent of the Pauli realization and relies only on the graph structure of the model.
	In Ref.~\cite{fendley2019free}, Fendley suggests applying this solution method to the cooper pair model of Refs.~\cite{Fendley2007Cooper,vanVoorden2019topological}, which represents one other such fermion model. 
	However, the frustration graph of this model has even holes, and so is ineligible for solution by our method. 
	Nevertheless, it would be interesting to investigate our method as a starting point for approximate solutions to non-integrable models such as quantum impurity models~\cite{Bravyi2017Complexity}. 
	One potential application would be to extend the exact analysis of Fendley's four-fermion model to an approximate one on periodic boundary conditions. 
	We leave such questions for future work.

	\begin{acknowledgements}
		We thank Paul Fendley and Sergey Bravyi for initially pointing us to Ref.~\cite{fendley2019free}. 
		We also thank Ryan Mann for enlightening discussions on graph theory and independence polynomials early in this project. 
		Finally, we thank Maria Chudnovsky and Paul Seymour for providing additional details to us about the algorithm in Ref.~\cite{Chudnovsky2012Growing} for finding a simplicial clique in a claw-free graph in polynomial time, and for sharing the result from Ref.~\cite{CSSS} that every ECF graph has a simplicial clique. 
		This research was supported by the Australian Research Council Centre of Excellence for Engineered Quantum Systems (EQUS, CE170100009).
	\end{acknowledgements}

\end{document}